\documentclass[twoside]{article}

\usepackage[accepted]{aistats2023}
\usepackage[round]{natbib}

\usepackage{mathtools}
\usepackage{amsthm} 
\usepackage{amsfonts}       
\usepackage{graphics} 
\usepackage{algorithm}
\usepackage{algorithmic}
\usepackage{subfig}
\usepackage{url}
\usepackage{verbatim}

\usepackage[utf8]{inputenc} 
\usepackage[T1]{fontenc}      
\usepackage{url}            
\usepackage{booktabs}       
\usepackage{bm}
\usepackage{amssymb}
\usepackage{pgf, tikz}
\usepackage{nicefrac}       
\usepackage{microtype}      
\usepackage{xcolor}         
\usetikzlibrary{arrows,automata,fit}

\usepgflibrary{shapes.geometric}
\newcommand{\xx}{1}
\newcommand{\yy}{1}
\newcommand{\stage}[3]{\tikz{\node[shape=circle,draw,inner sep=1pt,fill=#1]{$#3_{#2}$};}}

\newcommand{\stages}[2]{\tikz{\node[shape=circle,draw,inner sep=1pt,fill=#1,minimum size=0.5cm]{\scriptsize{$v_{#2}$}};}} 

\newcommand{\leaf}{\tikz{\node[shape=circle,draw,inner sep=1.5pt,fill=white] {};}}
\newcommand\independent{\protect\mathpalette{\protect\independenT}{\perp}}
\def\independenT#1#2{\mathrel{\rlap{$#1#2$}\mkern2mu{#1#2}}}

\newcommand{\Do}{\operatorname{do}}
\newcommand{\CID}{\operatorname{CID}}

\newtheorem{definition}{Definition}
\newtheorem{proposition}{Proposition}

\begin{document}

%

%

\twocolumn[

\aistatstitle{Context-Specific Causal Discovery for Categorical Data Using Staged Trees}

\aistatsauthor{ Manuele Leonelli \And Gherardo Varando}

\aistatsaddress{School of Science and Technology \\ IE University, Madrid, Spain \And  Image Processing Laboratory \\ Universitat de València, València, Spain} ]

\begin{abstract}
Causal discovery algorithms aim at untangling complex causal 
relationships from data. 
Here, we study causal discovery and inference methods based on staged tree models, which can represent complex and asymmetric causal relationships between categorical variables. We provide a first graphical representation of the equivalence class of a staged tree, by looking only at a specific subset of its underlying independences.
We further define a new pre-metric, inspired by the widely used structural intervention distance, to quantify the closeness between two staged trees in terms of their corresponding causal inference statements. A simulation study highlights the efficacy of staged trees in uncovering complexes, asymmetric causal relationships from data, and real-world data applications illustrate their use in practical causal analysis.
\end{abstract}

\section{INTRODUCTION}
One of the major tasks in all areas of science is to uncover causal relationships between variables of interest. Since experimental data is in many cases unavailable, this task often comes down to discover such relationships using  observational data only. This is usually referred to as \emph{causal discovery}. One of the most common approaches in causal discovery, as well as in causal analysis, is to represent causal relationships via \emph{directed acyclic graphs} (DAGs).  If there is an edge pointing from one variable to another in such graphs, then the former is a
direct cause of the latter~\citep{Pearl2009}. The literature on causal discovery using DAGs is now extensive~\citep[see][for a review]{Glymour2019}. The most common approaches are the PC-algorithm~\citep{Spirtes2000}, greedy equivalence search~\citep{Chickering2002} and functional 
causal models~\citep{Hoyer2008}.

Although more efficient and scalable causal discovery algorithms are still being developed \citep[e.g.,  ][]{bhattacharya2021differentiable,Monti2020}, most of the recent literature has focused on continuous random variables only. Attention to causal discovery for observational discrete data has been  limited \citep[see e.g.][for exceptions]{Cai2018,Cowell2014,Huang2018,Peters2010}. The aim of this paper is to discuss flexible and powerful causal discovery algorithms for categorical data embedding complex and asymmetric variable relationships and to highlight their efficacy. 

Whilst most causal discovery is carried out via DAG models, here we consider staged tree models~\citep{Collazo2018,Smith2008} which, differently to DAGs, can represent a wide array of asymmetric causal effects between categorical variables. Bayesian MAP structural learning algorithms for this model class have been introduced~\citep{Collazo2016,Cowell2014,Freeman2011} as well as score-based ones~\citep{Leonelli2021,Silander2013}. A 
wide selection of score-based algorithms have been implemented in the open-source \texttt{stagedtreees} R package~\citep{Carli2020} and are used henceforth.
Other strategies for non-symmetric relationships in DAG have been proposed in the literature. Two main approaches are modeling CPTs with tree 
structures~\citep{Chickering1997, Boutilier1996, Pensar2016} or using labelled graphs~\citep{Pensar2015}.

Despite the importance of uncovering causality from data, only one causal discovery algorithm for staged trees has been proposed \citep{Cowell2014}. Here we provide a suite of 
discovery algorithms based on the dynamic programming approach of \cite{Silander2013}, 
all freely available in the \texttt{stagedtrees} R package. 
Furthermore, we perform an extensive simulation study to assess their effectiveness, demonstrating that staged trees are extremely powerful in discovering complex dependence structures and in general outperform DAGs for categorical data.

Just as with DAGs, causal discovery with staged trees can be effectively carried out only if coupled with a method to compute the statistical equivalence class of a model \citep{Collazo2018}. However, the construction of this equivalence class has been shown to be extremely challenging \citep{duarte2021,Gorgen2018,gorgen2018discovery},  and no practical implementations are available. Here we provide a first graphical criterion to characterize part of the equivalence class of a staged tree by looking only at its symmetric independences and showcase its use in practice in our data applications. Approaches restricting the types of independences to be considered when studying equivalence have lately become popular also for DAG models \citep{markham2022transformational,textor2015,wienobst2020recovering}. Notice that once a causal staged tree model is chosen, there is a wide array of methods to estimate causal effects \citep{Genewein2020,Gorgen2015,Thwaites2010,thwaites2013causal}.

The quality of our routines is investigated by computing a new measure of dissimilarity between causal models tailored to the topology of staged trees and inspired by the widely-used structural intervention distance (SID) \citep{Peters2015} which we henceforth call \emph{context-specific intervention discrepancy} (CID). Differently from SID which only accounts for symmetric causal relationships, our defined CID can more generally consider the difference between two causal models by accounting for complex, asymmetric dependencies.

Summarizing, our 
contributions are the following: (i) a first graphical criterion of equivalence in staged trees; (ii) the first causal measure to compare asymmetric causal relationships in both staged trees and DAGs; (iii) a comparative simulation study highlighting the effectiveness of asymmetric causal discovery; (iv) multiple real-world data applications showcasing our methodology in practice.
The code with the implemented methods and the simulation experiments is available in the \texttt{stagedtrees} R package~\citep{Carli2020} and in the repository available at  \url{https://github.com/gherardovarando/stagedtrees_causal}.

\section{STAGED TREES}

Let  $[p]=\{1,\dots,p\}$ and $\bm{X}=(X_i)_{i\in[p]}$ be categorical random variables with joint mass function $P$ and sample space $\mathbb{X}=\times_{i\in[p]}\mathbb{X}_i$. For $A\subset [p]$, we let $\bm{X}_A=(X_i)_{i\in A}$ and $\bm{x}_A=(x_i)_{i\in A}$ where $\bm{x}_A\in\mathbb{X}_A=\times_{i\in A}\mathbb{X}_i$. We also let $\bm{X}_{-A}=(X_i)_{i\in [p]\setminus A}$.

Let $(V,E)$ be a directed, finite, rooted tree with vertex set $V$, root node $v_0$, and edge set $E$. 
For each $v\in V$, 
let $E(v)=\{(v,w)\in E\}$ be the set of edges emanating
from $v$ and $\mathcal{C}$ be a set of labels. 

\begin{definition}
\label{def:x}
An $\bf X$-compatible staged tree 
is a triple $(V,E,\eta)$, where $(V,E)$ is a rooted directed tree and:
\begin{enumerate}
    \item $V = {v_0} \cup \bigcup_{i \in [p]} \mathbb{X}_{[i]}$;
		\item For all $v,w\in V$,
$(v,w)\in E$ if and only if $w=\bm{x}_{[i]}\in\mathbb{X}_{[i]}$ and 
			$v = \bm{x}_{[i-1]}$, or $v=v_0$ and $w=x_1$ for some
$x_1\in\mathbb{X}_1$;
\item $\eta:E\rightarrow \mathcal{L}=\mathcal{C}\times \cup_{i\in[p]}\mathbb{X}_i$ is a labelling of the edges such that $\eta(v,\bm{x}_{[i]}) = (\kappa(v), x_i)$ for some 
			function $\kappa: V \to \mathcal{C}$. 
\end{enumerate}
	If $\eta(E(v)) = \eta(E(w))$ then $v$ and $w$ are said to be in the same 	\emph{stage}.
\end{definition} 

Therefore, the equivalence classes induced by  $\eta(E(v))$
form a partition of the internal vertices of the tree  in \emph{stages}.

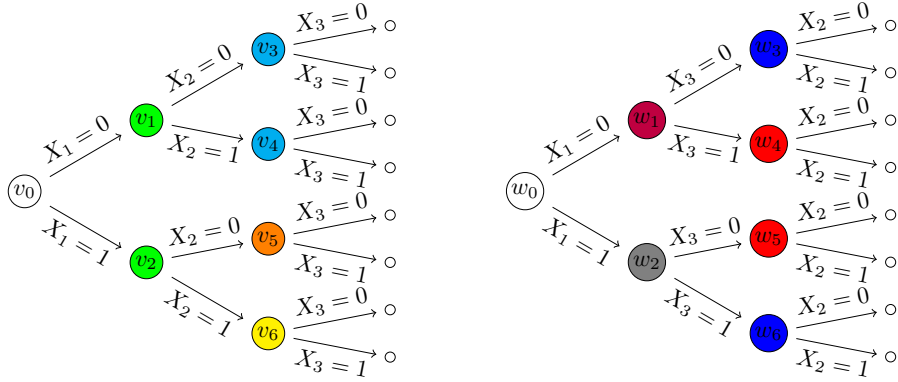
\begin{figure*}
\centering
\scalebox{0.9}{
\begin{tikzpicture}
\renewcommand{\xx}{1.8}
\renewcommand{\yy}{0.7}
\node (v1) at (0*\xx,0*\yy) {\stage{white}{0}{v}};
\node (v2) at (1*\xx,1.5*\yy) {\stage{green}{1}{v}};
\node (v3) at (1*\xx,-1.5*\yy) {\stage{green}{2}{v}};
\node (v4) at (2*\xx,3*\yy) {\stage{cyan}{3}{v}};
\node (v5) at (2*\xx,1*\yy) {\stage{cyan}{4}{v}};
\node (v6) at (2*\xx,-1*\yy) {\stage{orange}{5}{v}};
\node (v7) at (2*\xx,-3*\yy) {\stage{yellow}{6}{v}};
\node (l1) at (3*\xx,3.5*\yy) {\leaf};
\node (l2) at (3*\xx,2.5*\yy) {\leaf};
\node (l3) at (3*\xx,1.5*\yy) {\leaf};
\node (l4) at (3*\xx,0.5*\yy) {\leaf};
\node (l5) at (3*\xx,-0.5*\yy) {\leaf};
\node (l6) at (3*\xx,-1.5*\yy) {\leaf};
\node (l7) at (3*\xx,-2.5*\yy) {\leaf};
\node (l8) at (3*\xx,-3.5*\yy) {\leaf};
\draw[->] (v1) -- node [above, sloped] {$X_1=0$} (v2);
\draw[->] (v1) -- node [below, sloped] {$X_1=1$} (v3);
\draw[->] (v2) --  node [above, sloped] {$X_2=0$} (v4);
\draw[->] (v2) -- node [below, sloped] {$X_2=1$} (v5);
\draw[->] (v3) -- node [above, sloped] {$X_2=0$} (v6);
\draw[->] (v3) -- node [below, sloped] {$X_2=1$} (v7);
\draw[->] (v4) -- node [above, sloped] {$X_3=0$} (l1);
\draw[->] (v4) -- node [below, sloped] {$X_3=1$} (l2);
\draw[->] (v5) -- node [above, sloped] {$X_3=0$} (l3);
\draw[->] (v5) -- node [below, sloped] {$X_3=1$} (l4);
\draw[->] (v6) -- node [above, sloped] {$X_3=0$} (l5);
\draw[->] (v6) -- node [below, sloped] {$X_3=1$} (l6);
\draw[->] (v7) -- node [above, sloped] {$X_3=0$} (l7);
\draw[->] (v7) -- node [below, sloped] {$X_3=1$} (l8);
\end{tikzpicture}}
\hspace{1cm}
\scalebox{0.9}{
\begin{tikzpicture}
\renewcommand{\xx}{1.8}
\renewcommand{\yy}{0.7}
\node (v1) at (0*\xx,0*\yy) {\stage{white}{0}{w}};
\node (v2) at (1*\xx,1.5*\yy) {\stage{purple}{1}{w}};
\node (v3) at (1*\xx,-1.5*\yy) {\stage{gray}{2}{w}};
\node (v4) at (2*\xx,3*\yy) {\stage{blue}{3}{w}};
\node (v5) at (2*\xx,1*\yy) {\stage{red}{4}{w}};
\node (v6) at (2*\xx,-1*\yy) {\stage{red}{5}{w}};
\node (v7) at (2*\xx,-3*\yy) {\stage{blue}{6}{w}};
\node (l1) at (3*\xx,3.5*\yy) {\leaf};
\node (l2) at (3*\xx,2.5*\yy) {\leaf};
\node (l3) at (3*\xx,1.5*\yy) {\leaf};
\node (l4) at (3*\xx,0.5*\yy) {\leaf};
\node (l5) at (3*\xx,-0.5*\yy) {\leaf};
\node (l6) at (3*\xx,-1.5*\yy) {\leaf};
\node (l7) at (3*\xx,-2.5*\yy) {\leaf};
\node (l8) at (3*\xx,-3.5*\yy) {\leaf};
\draw[->] (v1) -- node [above, sloped] {$X_1=0$} (v2);
\draw[->] (v1) -- node [below, sloped] {$X_1=1$} (v3);
\draw[->] (v2) --  node [above, sloped] {$X_3=0$} (v4);
\draw[->] (v2) -- node [below, sloped] {$X_3=1$} (v5);
\draw[->] (v3) -- node [above, sloped] {$X_3=0$} (v6);
\draw[->] (v3) -- node [below, sloped] {$X_3=1$} (v7);
\draw[->] (v4) -- node [above, sloped] {$X_2=0$} (l1);
\draw[->] (v4) -- node [below, sloped] {$X_2=1$} (l2);
\draw[->] (v5) -- node [above, sloped] {$X_2=0$} (l3);
\draw[->] (v5) -- node [below, sloped] {$X_2=1$} (l4);
\draw[->] (v6) -- node [above, sloped] {$X_2=0$} (l5);
\draw[->] (v6) -- node [below, sloped] {$X_2=1$} (l6);
\draw[->] (v7) -- node [above, sloped] {$X_2=0$} (l7);
\draw[->] (v7) -- node [below, sloped] {$X_2=1$} (l8);
\end{tikzpicture}}

\caption{An example of an $(X_1, X_2, X_3)$-compatible (left) 
and an $(X_1, X_3, X_2)$-compatible (right) staged trees.
\label{fig:staged1}}
\end{figure*}

Definition \ref{def:x} first constructs a rooted tree where each root-to-leaf path, or equivalently each leaf, is associated with an element of the sample space $\mathbb{X}$.  Then a labeling of the edges of such a tree is defined where labels are pairs with one element from a set $\mathcal{C}$ and the other from the sample space $\mathbb{X}_i$ of the corresponding variable $X_i$ in the tree. By construction, $\bf X$-compatible staged trees are such that two vertices can be in the same stage if and only if they correspond to the same sample space. Although staged trees can be more generally defined without imposing this condition, henceforth, and as common in practice, we focus on $\bf{X}$-compatible staged trees only \citep[see][for an example of a non $\bf{X}$-compatible tree]{Leonelli2019}. 

Figure~\ref{fig:staged1} (left) reports an $(X_1,X_2,X_3)$-compatible stratified staged tree over three binary variables. The \textit{coloring} given by the function $\kappa$ is shown in the vertices and
each edge $(\cdot , (x_1, \ldots, x_{i}))$ is labeled with $X_i = x_{i}$. 
The edge labeling $\eta$ can be read from the graph combining the text label and the 
color of the emanating vertex. 
The staging of the staged tree in Figure~\ref{fig:staged1} is given by the partition $\{v_0\}$, $\{v_1,v_2\}$, $\{v_3,v_4\}$, $\{v_5\}$ and $\{v_6\}$.

The parameter space associated to an $\bf X$-compatible staged tree $T = (V, E, \eta)$ 
with 
labeling $\eta:E\rightarrow \mathcal{L}$ 
is defined as
\begin{align*}
\label{eq:parameter}
	\Theta_T=\Big\{\bm{\theta}\in\mathbb{R}^{\eta(E)} \;|\; \forall e\in E, \theta_{\eta(e)}\in (0,1) \textnormal{ and } \\
	\forall v\in V, \sum_{e\in E(v)}\theta_{\eta(e)}=1\Big\}
\end{align*}

Let $\bm{l}_{T}$ denote the leaves of a staged tree $T$. Given a vertex $v\in V$, there is a unique path in $T$ from the root $v_0$ to $v$, denoted as $\lambda(v)$. For any path $\lambda$ in $T$, let $E(\lambda)=\{e\in E: e\in \lambda\}$ denote the set of edges in the path $\lambda$.


\begin{definition}
\label{def:stmodel}
	The \emph{staged tree model} $\mathcal{M}_{T}$ associated to the $\bf X$-compatible staged 
	tree $(V,E,\eta)$ is the image of the map
\begin{equation}
\label{eq:model}
\begin{array}{llll}
\phi_T & : &\Theta_T &\to \Delta_{|\bm{l}_T| - 1} \\
 &  & \bm{\theta} &\mapsto \Big(\prod_{e\in E(\lambda(l))}\theta_{\eta(e)}\Big)_{l\in \bm{l}_T}
\end{array}
\end{equation}
\end{definition}

An element of $\mathcal{M}_T$ in Definition~\ref{def:stmodel} identifies a joint probability $P_{\bm{\theta}}$ with conditional distributions, for all $\bm{x}\in\mathbb{X}$ and  $i\in[p]$,
\[ 
 P_{\bm{\theta}}(X_i = x_i | X_{[i-1]} = \bm{x}_{[i-1]}) = \theta_{\eta( \bm{x}_{[i-1]}, \bm{x}_{[i]})}. 
 \] 
 
 \begin{definition}
Two staged trees $T$ and $S$ are said to be \emph{statistically equivalent}
if they induce the same 
models, that is $\mathcal{M}_T = \mathcal{M}_S$. 
\end{definition}

\subsection{Conditional Independence and graphical representation}

A symmetric or total, conditional independence statement, or just conditional independence (CI) 
($X_A \independent X_B | X_C$) holds for a 
probability distribution $P$, over categorical 
random variable, if 
\begin{equation}
    P(X_A| X_B = x_B, X_C = x_C) =  P(X_A | X_C = x_c),
    \label{eq:CI}
\end{equation}
for every $x_B \in \mathbb{X}_B$ and $x_c \in \mathbb{X}_C$.
Conditional independence statements can
be efficiently represented by DAGs model and 
the d-separation criterion~\citep{pearl1987logic, VERMA199069}.
In particular if a probability distribution $P$ over $X_1, 
\ldots, X_p$ belongs to 
the model class associated with a DAG $G$, then 
$X_i \independent X_j | X_C$ with respect to $P$ if 
$i$ and $j$ are d-separated by $C$ in $G$. Thus we can 
graphically read in $G$ the conditional independence statements
that hold for every distribution Markov with respect to $G$.

For categorical random variables, we could envisage those equality relationships such as 
Equation~\ref{eq:CI} hold 
only for a subset of values of $x_B$ and/or 
$x_C$. This generalized \textit{asymmetric} 
conditional independences
can be organized into three classes:
(i) Context-specific CI~\citep{Boutilier1996}, when 
    $P(X_A | X_B = x_B, X_C=x_C) = 
    P(X_A | X_C=x_c)$ for all $x_B 
    \in \mathbb{X}_B$ and 
    for a subset of possible value $x_C \in \mathcal{C} \subseteq \mathbb{X}_C$ (the context).
    (ii) Partial CI~\citep{Pensar2016}, when 
    $P(X_A |  X_B = x_B, X_C=x_C ) = 
    P(X_A | X_C=x_C)$
    for a subset of values $x_B \in \mathcal{B} \subseteq \mathbb{X}_B$ and a subset of 
    values $x_C \in \mathcal{C} \subseteq \mathbb{X}_C$.
    (iii) Local CI~\citep{Chickering1997}, when
    $P(X_A|  X_C = x^1_C) = 
    P(X_A | X_C = x^2_C)$.
Such asymmetric CI statements cannot be encoded 
graphically in a classical DAG model, since they refer to
equalities valid in specific conditional probability tables. 
Previous works have thus modeled such equality relationships by either modeling CPTs with tree structures~\citep{Chickering1997, Boutilier1996, Pensar2016}
or by considering only context-specific CIs and using 
labelled DAGs~\citep{Pensar2015}.

In the staged tree in Figure~\ref{fig:staged1} (left), 
we can see how the vertex staging represents conditional independence: the fact that $v_1$ and $v_2$ are in the same stage (green) implies that $X_1\independent X_2$;
in fact, from the definition of staged tree model, the
context-specific conditional distribution of $X_2$ given $X_1=0$, represented by the edges emanating from $v_1$, is equal to the conditional distribution of $X_2$ given $X_1=1$.  The staging given by the light-blue vertices implies instead the context-specific independence~\citep{Boutilier1996} $X_3\independent X_2|X_1=0$: the independence between $X_3$ and $X_2$ holds only for one of the two levels of $X_1$. For such a staged tree there is no equivalent DAG representation since it embeds non-symmetric conditional independences~\citep{Varando2021}.

In the staged tree in Figure~\ref{fig:staged1} (right), we can
observe that the stages structure for the last variable 
implies the following equalities: 
$P(X_2| X_1 = 0,\, X_3=0) = P(X_2| X_1 = 1,\, X_3=1)$
and $P(X_2| X_1 = 0,\, X_3=1) = P(X_2| X_1 = 1,\, X_3=0)$. 
This is what is defined as a local CI~\citep{Chickering1997},
it is a relationship between conditional probabilities which 
cannot be expressed as traditional conditional independence nor 
context-specific or partial. 
We refer to \cite{Pensar2016} and \cite{Varando2021} for  additional discussion and examples of 
asymmetric CIs.

\subsection{Staged Trees and DAGs}
Consider a DAG $G$ and the associated statistical model $\mathcal{M}_G$ of all distributions that are Markov to $G$. \citet{Smith2008} showed that one can always construct a staged tree $T_G$ such that $\mathcal{M}_G=\mathcal{M}_{T_G}$. 
However, given a staged tree $T$ in general one cannot find a DAG $G_T$ such that $\mathcal{M}_T=\mathcal{M}_{G_T}$ since staged trees embed asymmetric independences that DAGs cannot represent. 

\citet{Varando2021} demonstrated that it is possible to find a minimal DAG $G_T=([p],F)$ such that 
$\mathcal{M}_T\subseteq\mathcal{M}_{G_T}$, and this minimal $G_T$ represents all symmetric conditional independences of $T$. More formally,  $X_i \independent X_i|\bm{X}_C$ holds in $\mathcal{M}_T$ if and only if $i$ and $j$ are d-separated by $C$ in $G_T$. For instance, the minimal DAG representation of the staged tree in 
Figure~\ref{fig:staged1} (left) is the v-structure $1\rightarrow 3 \leftarrow 2$.

In \citet{Varando2021} DAGs $G_T$ are also extended to have a labeling of their edges according to the type of dependence existing between any pair of random variables in the underlying staged tree $T$ \citep[according to the categorization of asymmetric independence given in][]{Pensar2016}. They termed such labeled DAGs as asymmetry-labeled DAGs (ALDAGs) and introduce algorithms to learn them from data. For the purposes of this paper, we are interested in a simplified version 
of the labeling; we are, in particular,
interested only in identify 
the subset of edges in $G_T$ which cover some 
asymmetric conditional independence statements. 

\begin{definition}
Let $G_T$ be the minimal DAG of an $\bm{X}$-compatible staged tree $T=(V,E,\eta)$ and 
an edge $(i, j)$ of $G$ is 
called \textit{non-total} if
\begin{align*}
\eta(E(\bm{x}_{[j-1]})) = \eta(E(\bm{x}'_{[j-1]})),\\
\text{for } \bm{x}_{[j-1]}, \bm{x}'_{[j-1]} \in 
\mathbb{X}_{[j-1]},  \\ 
\,\text{s.t.}\, 
{x}_i \neq x'_i.
\end{align*}
An edge is called \textit{total} otherwise. 
\end{definition}

Intuitively an edge $(i,j)$ in a minimal DAG is non-total 
when the variable $X_i$  is ``involved'' is a non-symmetrical conditional
independence for $X_j$.
The edges of a minimal DAG $G_T=([p],F)$ can thus be partitioned in total ($F^{tot}$)
and non-total edges ($F^{nt}$). 

As an illustration, in the 
minimal DAG from the tree on the left of 
Figure~\ref{fig:staged1} ($1\rightarrow 3 \leftarrow 2$), the edge $(2,3)$ is non-total because  $v_3$ and $v_4$ belong 
to the same stage.

\subsection{Causal Models Based on Staged Trees}

We can define a 
finite-interventional causal model~\citep{rischel21a}
induced by a staged tree 
as a collection of interventional
distributions in an intuitive way:  the joint distribution of $\bm{x}$ 
in the intervened model is obtained by the product of the parameters in the corresponding root-to-leaf path where we replace parameters corresponding to intervened variables.

\begin{definition}
    \label{def:causalstagedtree} 
    A staged tree causal model  
    induced by an $\mathbf{X}$-compatible staged tree 
    $T = (V, E, \eta)$ is the class of interventional 
    distributions defined, for each parameter vector
    $\bm{\theta} \in \Theta_T$, as follows:
    \begin{equation}
    \small
    \begin{aligned}
    &P_{\bm{\theta}}\left(\bm{X} = \bm{x} | \Do( \bm{X}_I = \bm{z}_I) \right) = 
    \prod_{i \not\in I} \theta_{\eta(\bm{x}_{[i]}, \bm{x}_{[i-1]})}  \prod_{i \in I} \delta(x_k, z_k) \\ 
    &= 
    \left\{ 
    \begin{matrix} 
    \frac{ P_{\bm{\theta}}( \bm{X} = \bm{x} ) }{\prod_{i \in I} 
    P_{\bm{\theta}}(X_i = x_i | \bm{X}_{[i-1]} = \bm{x}_{[i-1]}) }  & \text{if } \bm{x}_I = \bm{z}_I\\ 
    0   & \text{otherwise}
    \end{matrix} \right.
    \end{aligned}
    \end{equation}
\end{definition}

In particular, under the empty intervention, we recover the observational distribution $P_{\bm{\theta}}$. 

We say that two staged trees $T$ and $S$ are causally equivalent if 
they induce the same class of interventional distributions as in Definition~\ref{def:causalstagedtree}.
Obviously, two causally equivalent staged trees are also statistically equivalent but
not vice versa. 

We have that, as for DAGs, intervening on some variables, only affects
downstream variables, for $i \not\in I$ 
and an $\bm{X}$-compatible staged tree:
\begin{equation}
\begin{aligned}
P_{\bm{\theta}}\left(X_i | \Do(\bm{X}_I = \bm{x}_I)\right) = 
P_{\bm{\theta}}\left(X_i | \Do(\bm{X}_{I^*} = \bm{x}_{I^*}) \right),\\
\text{where } I^* = I \cap [i-1]
\end{aligned}
\end{equation}

And, in particular,
\begin{align*}
P_{\bm{\theta}}\left( X_i = x_i | \Do(\bm{X}_{[i-1]} = \bm{x}_{[i-1]}) \right) &=  \\
P_{\bm{\theta}}\left( X_i = x_i | \bm{X}_{[i-1]} = \bm{x}_{[i-1]} \right) &= \theta_{\eta( \bm{x}_{[i-1]}, \bm{x}_{[i]})}.
\end{align*}

\section{CAUSAL DISCOVERY ALGORITHMS}
\label{sec:methods}

As discussed by~\cite{Collazo2018} and~\cite{Cowell2014}, causal discovery algorithms for staged trees must combine two routines: (i) an algorithm learning the stage structure of the tree with a fixed variable ordering; (ii) an algorithm exploring the possible variable orderings. Both are reviewed next.

\subsection{Learning the Stage Structure with a Fixed Order}
\label{sec:stages}
The space of possible $\bm{X}$-compatible staged trees is considerably larger than the space of possible DAGs. 
Even if we fix the order of the variables, exploring all possible combinations of 
stages structure becomes rapidly infeasible~\citep{Collazo2018}. 
We thus use two of the possible heuristic 
searches implemented in the 
\texttt{stagedtrees} package~\citep{Carli2020}. In both cases, we use the BIC score as criterion for selecting the 
best ordering of the variables~\citep[see][for details]{Gorgen2020}.  However, our implementation can be coupled with any algorithm available in the \texttt{stagedtrees} package. 

The backward hill-climbing (BHC) method 
consists in
starting from the saturated model and, 
for each variable, iteratively 
trying to join stages. 
At each step of the algorithm, all possible combinations of two stages are tried and the best move is chosen. 
Since the log-likelihood decomposes across the depth of the tree, the stages search can be performed  independently for each variable. 

The use of the k-means clustering of probabilities to learn staged event tree was first introduced by~\citet{Silander2013} as a fast alternative to 
backward hill-climbing algorithms. 
The default version implemented in the \texttt{stagedtrees} package performs k-means clustering over the square root of the probabilities of a given variable given all the possible contexts. 
Both algorithms operate over the stage structures of each variable $X_i$ independently of the other 
variable stages. Furthermore, the estimated stage structure of a given variable depends only on which variables precede $X_i$, independently of their order.

\subsection{Learning an Optimal Variable Ordering}
\label{sec:order}
The methods described in the previous section
output an $\bm{X}_{\pi}$-compatible staged tree for a possible ordering $\pi$ of the variables.  
For a small number of variables, it is possible to 
simply enumerate all possible $p!$ staged  trees for all possible orders, and select the best one(s) according to a chosen criterion  (e.g. BIC). \citet{Silander2013} proposed a dynamic 
programming algorithm that still obtains a 
global optimum, but with a substantial 
reduction in computational complexity. 
The method can be coupled with 
every algorithm which operates independently on
every variable and using as guiding score 
any function which can be decomposed 
across the variables of the model.

\subsection{Related Work}

In principle, non-symmetric CI statements are  
represented by equalities in conditional probability tables (CPT) in categorical DAG parametrizations. 
Still, a classical 
(or full tables) DAG is not able to represent 
\textit{graphically} such asymmetric relationships, in the sense that such equalities
are not encoded in any particular structure. 
A simple extension of DAGs could consider
additional nodes representing 
values of variables (e.g. $X_1 = 0$, $X_1 = 1$) in 
order to represent context-specific relationships. 
Unfortunately, this strategy
 works only for univariate contexts and, 
 it would entail
deterministic relationships between some nodes in the DAG (e.g. $X_1=0$ and $X_1 = 1$).

More complex strategies for non-symmetric relationships in DAGs have been proposed in the literature. Two main
approaches are modeling CPTs with tree structures~\citep{Chickering1997, Boutilier1996, Pensar2016}
or use labelled graphs~\citep{Pensar2015} for context-sepcific independences. 
DAGs with tree-parametrized CPTs and staged tree methods are very similar approaches that use trees
to represent conditional probabilities. 
In particular, the statistical models represented by 
staged tree and DAGs with tree-CPTs are, in principle, 
equivalent. 
Even the learning algorithm proposed by \citet{Pensar2016} 
consist in a heuristic search using splitting and joining 
operation on each CPT-tree, similar in a way to the 
hill-climbing moves proposed also for staged trees~\citep{Carli2020}.
The difference is that, in the
staged tree approach, we do not assume a sparse DAG between variables and we do not search both a DAG and sparse
CPTs. 
Of course, restricting to sparse DAG is beneficial from a computational perspective, and it has been proposed and
shown to be effective also for staged trees~\citep{Barclay2013, pmlr-v186-leonelli22a}. Unfortunately, we are not aware of any available implementation of these
related methods and we were thus 
unable to run any empirical comparisons.

\subsection{Exploring the Equivalence Class}

Given a learned staged tree from data, any formal causal analysis also needs exploration of the associated statistical equivalence class. For staged trees, this has been shown to be extremely complex. \citet{Gorgen2018} and \citet{gorgen2018discovery} give polynomial criteria which are complex to implement in practice, whilst~\citet{duarte2021} considers a particular subclass of staged trees. The following proposition paves the way toward the exploration of the equivalence class of a staged tree.

\begin{proposition}
\label{prop:equivalence}
Let $T$ be an $\bm{X}$-compatible staged tree and 
$G_T=([p], 
F=F^{tot} \cup F^{nt})$ 
its minimal DAG,
where we denote with $F^{nt}$ the non-total edges of $G_T$. 
Let $G'=([p], F' \cup F^{nt})$ be a DAG in the 
same Markov equivalence class of $G_T$, where $\pi$ is one of its topological orders.
If additionally $([p], F')$ and $([p], F^{tot})$ are 
Markov equivalent,
then there exists an $\bm{X}_{\pi}$-compatible staged tree $S$ (with minimal DAG $G'$)
such that $\mathcal{M}_T=\mathcal{M}_{S}$.
Vice versa, 
If $\mathcal{M}_T=\mathcal{M}_{S}$
their minimal DAGs $G_T$ and $G_{S}$ are Markov equivalent.
\end{proposition}

However, 
there may be equivalent staged trees whose 
minimal DAGs
have ``non-total" edges with a different directionality~\citep[see e.g.][]{Pensar2015}. 
As an illustration of Proposition~\ref{prop:equivalence}, consider the staged tree in Figure~\ref{fig:staged1} (left). Its 
minimal DAG
is the v-structure $X_1\rightarrow X_3 \leftarrow X_2$. 
Therefore there exists, at least, an 
$(X_2,X_1,X_3)$-compatible staged tree which is statistically equivalent to the one in 
Figure~\ref{fig:staged1} (left), and  there cannot be a statistically equivalent staged tree where $X_3$ is not the last variable.

Although Proposition \ref{prop:equivalence} does not give a complete characterization of the equivalence class, it is important because it informs about relationships existing in the staged tree which cannot be interpreted as causal if learned from data. We showcase in Section \ref{sec:data} how the proposition can be used for applied causal analyses.

\section{CONTEXT-SPECIFIC INTERVENTIONAL DISCREPANCY}

Similarly to the structural interventional distance (SID) for DAGs~\citep{Peters2015}, 
which counts the number of wrongly estimated interventional distributions, we can define a context interventional discrepancy, with respect to 
a reference staged tree ($T$).
Such a discrepancy measures the extent of the errors done in 
computing context-specific interventions using 
a different staged tree ($S$).

\begin{definition}
\label{def:cid}
Let $T = (V, E, \eta)$ be an $\bm{X}$-compatible staged event tree
and $S = (W, F, \nu)$ an $\bm{X}_{\pi}$-compatible staged event tree,
where $\pi$ is a permutation 
of  $[p]$. 
We define the context interventional discrepancy $\CID(T, S)$ as,
\[ \CID(T, S) =  \sum_{i \in [p]} \CID_i(T, S), \] 
where $\CID_i\left(T,S\right)$ is the
proportion of contexts $\bm{x}_{[i-1]} \in \mathbb{X}_{[i-1]}$
for which the
interventional distribution 
$P(X_{i} | \Do(\bm{X}_{[i-1]} = \bm{x}_{[i - 1]}))$
is wrongly inferred by $S$ with respect to $T$.
Precisely, we say that $P(X_{i} | \Do(\bm{X}_{[i-1]} = \bm{x}_{[i-1]}))$ 
is wrongly inferred by $S$ with respect to $T$ if 
there exists $P \in \mathcal{M}_T$ such that 
\begin{multline*}  
 P(X_{i} | \bm{X}_{[i-1]} = \bm{x}_{[i - 1]}) \neq  \\
  P\left( X_i| \bm{X}_{I} \in  \left\{ 
  \begin{matrix}
   \bm{y}_{I} \in \mathbb{X}_{I} :
  \nu(E(\bm{y}_{K})) = \nu(E(\bm{x}_K))  \\
  \text{ for some } \bm{y}_{K \setminus I}, 
  \bm{x}_{K \setminus I} \in \mathbb{X}_{K\setminus I}
  \end{matrix}
 \right\} \right)  
\end{multline*}
where 
$K = \{ j : \pi^{-1}(j) < \pi^{-1}(t) \}$ are the variables preceding $X_i$ in $S$ and 
$I = K \cap [i-1]$.
\end{definition}


Intuitively, CID measures how much a different staged tree $S$ can be used to compute 
interventional distributions of the type 
$P(X_{i} | \Do(\bm{X}_{[i-1]} = \bm{x}_{[i - 1]}))$. 
We choose to consider only univariate 
distributions under interventions on the 
preceding variables in the true model $T$, instead of pairwise interventional distributions of the form $P(X_i | \Do(X_j = x_j))$ as in the definition of SID~\citep{Peters2015}. This is because our goal is
to quantify the effect of context-specific interventions. 
Moreover, differently from SID, we do not control 
if the additional variables preceding $X_i$ in $S$ 
are a valid adjustment set. This simplification
is taken to avoid the computational complexity of having 
to check the stages structure for all the variables in the 
eventual adjustment set. 
Nevertheless, the proposed CID is a sensitive measure of
correctness of the causal model, as Proposition~\ref{prop:cid} 
and Figure~\ref{fig:density} show.
However, other measures of the differences 
between staged tree causal models could be alternatively defined, eventually considering different intervention classes or
validity of the adjustment sets. 

The algorithm to compute the context-specific 
interventional discrepancy
is given in the Appendix where its correctness is also proven.

As an example of the computation of CID$(T,S)$, consider the two staged trees in Figure~\ref{fig:staged1}, 
where the left tree is $T$ and the right one is $S$. We need to determine which interventional distributions for the left staged tree are wrongly inferred by the right one.
For example, consider the intervention $\Do(X_1 = 0 , X_2  = 1)$ and the distribution 
$P(X_3|\Do(X_1 = 0 , X_2  = 1))$ for $P \in \mathcal{M}_T$. We have that, 
\[ P(X_3|\Do(X_1 = 0 , X_2  = 1))  = P(X_3| X_1 = 0 , X_2  = 1 ),\]
and, $I = \{1\}$, thus, because $v_3$ and $v_4$ belong to the same stage in $T$:
\begin{align*}
 P(X_3| X_1 = 0) &= P(X_3|X_1 = 0, X_2 = 0)\\
 &=  P(X_3|X_1 = 0, X_2 = 1),
 \end{align*}
and $P(X_3|\Do(X_1 = 0 , X_2  = 1))$ is then correctly inferred by $S$.
On the other hand, we have that $P(X_3|\Do(X_1 = 1 , X_2  = 1))$ is wrongly inferred 
by $S$ because $P(X_3,|X_1 = 1 , X_2  = 1) \neq  P(X_3 | X_1 = 1)$ in general. To see this, notice that 
\begin{equation*}
\begin{aligned}
     P(&X_3 | X_1 = 1) = \\ 
     &\sum_{x_2 = 0,1}
        P(X_3 | X_1 = 1, X_2 = x_2)P(X_2 = x_2| X_1 = 1).
\end{aligned}
\end{equation*} 
And $P(X_3,|X_1 = 1 , X_2  = 1) \neq  P(X_3 | X_1 = 1)$ if, for example, we choose $P(X_2 = x_2| X_1 = 1) = 0.5$ and 
$P(X_3|X_1 = 1, X_2 = 0) \neq  P(X_3|X_1 = 1, X_2 = 1)$.

As another example, consider the interventional distribution
$P(X_2| \Do(X_1 = 1) ) $. In this case we have $I=\{1\}$, and
since vertices $v_1, v_2$  are in the same stage (and so are $\{ w_5, w_4\}$ and 
$\{ w_3, w_6\}$), we  have that 
\[P(X_2| X_1 = 1 )  = P(X_2) =  P(X_2 | X_1 \in \{ 0, 1 \} ), \]
that is, $S$ correctly infers the interventional distribution $P(X_2| \Do(X_1 = 1))$
for every $P \in \mathcal{M}_T$.

Similar to SID, the context-specific intervention discrepancy 
is not symmetric. The following proposition 
collects some properties of the newly defined 
measure. 

\begin{proposition}
\label{prop:cid}
The following properties hold for $\CID$. 
\begin{enumerate}
\renewcommand{\theenumi}{\roman{enumi}}
    \item $\CID(T,S) = 0$ for every pair of causally equivalent staged trees $S,T$.
    \item If $\mathcal{M}(T) \subseteq \mathcal{M}(S)$ and $\pi$ is the identity, then $\CID(T,S) = 0$.
    \item If $\mathcal{M}_T$ is the full independence model then $CID(T, S) = 0$ 
    for every $X_{\pi}$-compatible staged tree $S$.
\end{enumerate}
\end{proposition}

Notice that CID can also be used to compare categorical causal DAGs, since we can always transform a DAG $G$ to its equivalent staged tree representation $T_G$.
In order to compare CID and SID we perform a simulation study where we sample uniformly DAGs over $5$ binary variables and compute their CID and SID. 
The results are reported in the two-dimensional density and scatter plot in Figure~\ref{fig:density}. We can see that there is a high correlation between the two measures thus highlighting that CID is a sensible measure that could be used not only for non-symmetric models but also for symmetric ones based on DAGs.

\begin{figure}
    \centering
    \includegraphics[width=0.48\textwidth]{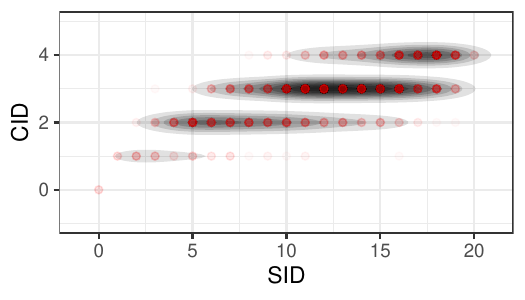}
    \caption{CID and SID between randomly generated DAGs over $5$ binary variables. The correlation between CID and SID is $0.67$ ($95\%$ confidence interval: $(0.633, 0.701)$)}.
    \label{fig:density}
\end{figure}

\section{SIMULATION EXPERIMENTS}
\label{sec:simulation}
We perform a simulation study to evaluate the feasibility of the proposed 
approach and to demonstrate its superiority with respect to the classical DAG algorithms under the assumption that the true model is a staged tree.  We simulate data from randomly generated staged tree models with different degrees of complexity: number of stages per variable ($k \in \{2,3,4\}$).
We consider models with $3, \ldots, 6$ binary variables, and sample sizes ranging from 
$100$ to $10000$ observations. 
For each parameters' combination, we perform $100$ repetition of the 
experiment each time randomly shuffling the order of the variables to 
eliminate any possible bias of the search heuristics. 

\begin{figure}
    \centering
    \includegraphics{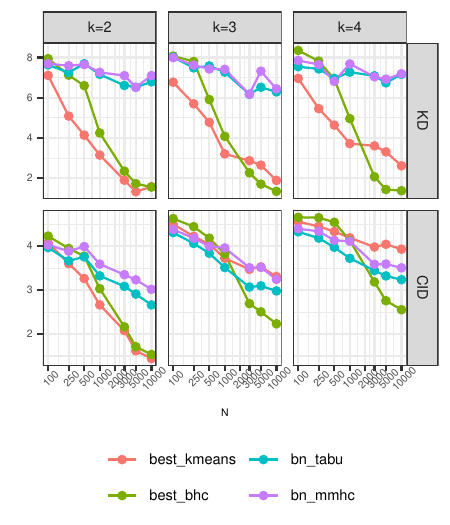}
    \caption{Context interventional discrepancy (\texttt{CID}) and Kendall tau distance (\texttt{KD}) between the estimated and true
    model.}
    \label{fig:exp}
\end{figure}

We run the staged trees approach described in Section~\ref{sec:methods} 
using the backward hill-climbing search (\texttt{best\_bhc}) and the
k-means heuristic (\texttt{best\_kmeans}) with the number of clusters 
fixed to $2$.  Our two routines are compared to two classical DAG learning algorithms 
such as tabu search (\texttt{tabu})~\citep{russell2009artificial}
and max-min hill-climbing (\texttt{mmhc})~\citep{tsamardinos2006max},
both implemented in the \texttt{bnlearn} R package~\citep{Scutari2010}. Results are displayed in Figure~\ref{fig:exp}, 
where the average CID and Kendall tau distance between the variable orderings are plotted as a function
of the sample size $N$ for the case of $6$ variables. The Kendall tau distance is computed between the true 
causal order and the estimated order with the implementation 
in the \texttt{PerMallows} R package~\citep{permallows}. 


We can observe that, as expected, methods based on staged trees are able to better recover 
 the causal structure of the true model. Since the true models are randomly generated staged trees we can expect that algorithms which search for the best DAG 
 are not able to recover the true relationships between the variables. 
  More specifically, we observe that the method based on the k-means algorithm works very well when the number of stages per variable matches the number of clusters ($k = 2$), while, with respect to CID, its performance degrades for $k = 3,4$. 
  The backward hill-climbing method, instead, requires a bigger sample size, but 
  it is able to perform well even when the true model is more complex. Additionally, it is interesting to notice that both staged tree approaches are able 
to recover well the causal order of the variables in all considered scenarios. This is especially interesting for the k-means algorithm which performs better than the 
backward hill-climbing method with respect to the Kendall distance even when it is misspecified ($k>2$). In the Supplementary Materials, we report the computational times of the algorithms. As expected, algorithms for DAGs are faster since the searched model space is much smaller. The k-means algorithm for staged trees is comparable to those for DAGs in terms of speed and its complexity does not seem to exponentially increase as in the case of the backward hill-climbing.

\section{REAL WORLD EXAMPLES}
\label{sec:data}

\subsection{ISTAT: Aspects on Everyday Life}
We illustrate the use of staged trees to uncover causal relationships using data from the 2014 survey ``\textit{Aspects on everyday life}" collected by ISTAT (the Italian National Institute of Statistics)~\citep{ISTAT}. The survey collects information from the Italian population on a variety of aspects of their daily lives. For the purpose of this analysis, we consider five of the many questions asked in the survey: do you practice sports regularly? (S = yes/no); do you have friends you can count on? (F = yes/no); do you trust people? (P = yes/no); 
are you satisfied with the 
environment situation of the area you live in? (E = yes/no, grouped from the original four levels);
do you watch TV? (T = yes/no, grouped from the original three levels). Instances with missing answers were dropped, resulting in 35870 answers to the survey.

We learn the 
staged structure 
and the variable ordering 
with the BHC algorithm 
coupled with the 
dynamical programming approach 
discussed in Section~\ref{sec:methods}.
The resulting 
tree is depicted in Figure~\ref{fig:staged3} where we can observe 
that the stages structure for the first three variables is 
equivalent to the conditional independence statement $S \independent E | F$.

Using Proposition~\ref{prop:equivalence} and
from 
the Markov equivalence class of the minimal DAG, 
and in particular of the sub-DAG S$\leftarrow$ F$\rightarrow$E,
it is 
easy to obtain that there are four equivalent orders of the first three variables (S-F-E;E-F-S;F-S-E;F-E-S) that give rise to statistically equivalent staged trees.

While the causal order among F,E and S cannot 
be completely recovered from data, 
the asymmetrical relationship 
in the stages structure of the last two 
variables (P and T in Figure~\ref{fig:staged3})
imply that there are no statistically 
equivalent staged trees where P or T 
appear before F,E or S~\citep[see][for a similar observation]{Gorgen2018}. 
Therefore the data support the hypothesis that 
F, E, and S affect whether an individual trusts people. 
Similarly, all previous variables appear to have a causal effect on whether an individual watches TV.

\begin{figure}
\centering
\scalebox{0.65}{
\begin{tikzpicture}
\renewcommand{\xx}{2}
\renewcommand{\yy}{1.9}
\node (v1) at (0*\xx,0*\yy) {\stages{white}{0}};
\node (v2) at (1*\xx,2*\yy) {\stages{pink}{1}};
\node (v3) at (1*\xx,-2*\yy) {\stages{purple}{2}};
\node (v4) at (2*\xx,3*\yy) {\stages{brown}{3}};
\node (v5) at (2*\xx,1*\yy) {\stages{brown}{4}};
\node (v6) at (2*\xx,-1*\yy) {\stages{gray}{5}};
\node (v7) at (2*\xx,-3*\yy) {\stages{gray}{6}};
\node (l1) at (3*\xx,3.5*\yy) {\stages{cyan}{7}};
\node (l2) at (3*\xx,2.5*\yy) {\stages{blue}{8}};
\node (l3) at (3*\xx,1.5*\yy) {\stages{blue}{9}};
\node (l4) at (3*\xx,0.5*\yy) {\stages{orange}{10}};
\node (l5) at (3*\xx,-0.5*\yy) {\stages{orange}{11}};
\node (l6) at (3*\xx,-1.5*\yy) {\stages{green}{12}};
\node (l7) at (3*\xx,-2.5*\yy) {\stages{orange}{13}};
\node (l8) at (3*\xx,-3.5*\yy) {\stages{violet}{14}};
\node (l9) at (4*\xx,3.75*\yy){\stages{brown}{15}};
\node (l10) at (4*\xx,3.25*\yy){\stages{yellow}{16}};
\node (l11) at (4*\xx,2.75*\yy){\stages{yellow}{17}};
\node (l12) at (4*\xx,2.25*\yy){\stages{cyan}{18}};
\node (l13) at (4*\xx,1.75*\yy){\stages{red}{19}};
\node (l14) at (4*\xx,1.25*\yy){\stages{red}{20}};
\node (l15) at (4*\xx,0.75*\yy){\stages{brown}{21}};
\node (l16) at (4*\xx,0.25*\yy){\stages{yellow}{22}};
\node (l17) at (4*\xx,-0.25*\yy){\stages{brown}{23}};
\node (l18) at (4*\xx,-0.75*\yy){\stages{brown}{24}};
\node (l19) at (4*\xx,-1.25*\yy){\stages{brown}{25}};
\node (l20) at (4*\xx,-1.75*\yy){\stages{yellow}{26}};
\node (l21) at (4*\xx,-2.25*\yy){\stages{red}{27}};
\node (l22) at (4*\xx,-2.75*\yy){\stages{red}{28}};
\node (l23) at (4*\xx,-3.25*\yy){\stages{red}{29}};
\node (l24) at (4*\xx,-3.75*\yy){\stages{brown}{30}};
\node (l25) at (5*\xx,3.875*\yy){\leaf};
\node (l26) at (5*\xx,3.625*\yy){\leaf};
\node (l27) at (5*\xx,3.375*\yy){\leaf};
\node (l28) at (5*\xx,3.125*\yy){\leaf};
\node (l29) at (5*\xx,2.875*\yy){\leaf};
\node (l30) at (5*\xx,2.625*\yy){\leaf};
\node (l31) at (5*\xx,2.375*\yy){\leaf};
\node (l32) at (5*\xx,2.125*\yy){\leaf};
\node (l33) at (5*\xx,1.875*\yy){\leaf};
\node (l34) at (5*\xx,1.625*\yy){\leaf};
\node (l35) at (5*\xx,1.375*\yy){\leaf};
\node (l36) at (5*\xx,1.125*\yy){\leaf};
\node (l37) at (5*\xx,0.875*\yy){\leaf};
\node (l38) at (5*\xx,0.625*\yy){\leaf};
\node (l39) at (5*\xx,0.375*\yy){\leaf};
\node (l40) at (5*\xx,0.125*\yy){\leaf};
\node (l56) at (5*\xx,-3.875*\yy){\leaf};
\node (l55) at (5*\xx,-3.625*\yy){\leaf};
\node (l54) at (5*\xx,-3.375*\yy){\leaf};
\node (l53) at (5*\xx,-3.125*\yy){\leaf};
\node (l52) at (5*\xx,-2.875*\yy){\leaf};
\node (l51) at (5*\xx,-2.625*\yy){\leaf};
\node (l50) at (5*\xx,-2.375*\yy){\leaf};
\node (l49) at (5*\xx,-2.125*\yy){\leaf};
\node (l48) at (5*\xx,-1.875*\yy){\leaf};
\node (l47) at (5*\xx,-1.625*\yy){\leaf};
\node (l46) at (5*\xx,-1.375*\yy){\leaf};
\node (l45) at (5*\xx,-1.125*\yy){\leaf};
\node (l44) at (5*\xx,-0.875*\yy){\leaf};
\node (l43) at (5*\xx,-0.625*\yy){\leaf};
\node (l42) at (5*\xx,-0.375*\yy){\leaf};
\node (l41) at (5*\xx,-0.125*\yy){\leaf};
\draw[->] (v1) --  node [above, sloped] {\scriptsize{F = yes}} (v2);
\draw[->] (v1) -- node [below, sloped] {\scriptsize{F = no}}(v3);
\draw[->] (v2) --  node [above, sloped] {\scriptsize{E = yes}}(v4);
\draw[->] (v2) --  node [below, sloped] {\scriptsize{E = no}}(v5);
\draw[->] (v3) --  node [above, sloped] {\scriptsize{E = yes}} (v6);
\draw[->] (v3) --  node [below, sloped] {\scriptsize{E = no }} (v7);
\draw[->] (v4) --  node [above, sloped] {\scriptsize{S = yes}} (l1);
\draw[->] (v4) -- node [below, sloped] {\scriptsize{S = no}}  (l2);
\draw[->] (v5) -- node [above, sloped] {\scriptsize{S = yes}}  (l3);
\draw[->] (v5) -- node [below, sloped] {\scriptsize{S = no}}  (l4);
\draw[->] (v6) -- node [above, sloped] {\scriptsize{S = yes}} (l5);
\draw[->] (v6) -- node [below, sloped] {\scriptsize{S = no}} (l6);
\draw[->] (v7) -- node [above, sloped] {\scriptsize{S = yes}} (l7);
\draw[->] (v7) -- node [below, sloped] {\scriptsize{S = no}} (l8);
\draw[->] (l1) --  node [above, sloped] {\scriptsize{P = yes}}(l9);
\draw[->] (l1) --  node [below, sloped] {\scriptsize{P = no}} (l10);
\draw[->] (l2) --  node [above, sloped] {\scriptsize{P = yes}}(l11);
\draw[->] (l2) --  node [below, sloped] {\scriptsize{P = no}} (l12);
\draw[->] (l3) --  node [above, sloped] {\scriptsize{P = yes}}(l13);
\draw[->] (l3) --  node [below, sloped] {\scriptsize{P = no}} (l14);
\draw[->] (l4) --  node [above, sloped] {\scriptsize{P = yes}}(l15);
\draw[->] (l4) --  node [below, sloped] {\scriptsize{P = no}} (l16);
\draw[->] (l5) -- node [above, sloped] {\scriptsize{P = yes}} (l17);
\draw[->] (l5) --  node [below, sloped] {\scriptsize{P = no}} (l18);
\draw[->] (l6) -- node [above, sloped] {\scriptsize{P = yes}} (l19);
\draw[->] (l6) --  node [below, sloped] {\scriptsize{P = no}} (l20);
\draw[->] (l7) -- node [above, sloped] {\scriptsize{P = yes}} (l21);
\draw[->] (l7) --  node [below, sloped] {\scriptsize{P = no}} (l22);
\draw[->] (l8) -- node [above, sloped] {\scriptsize{P = yes}} (l23);
\draw[->] (l8) -- node [below, sloped] {\scriptsize{P = no}} (l24);
\draw[->] (l9) -- node [above, sloped] {\scriptsize{T = yes}} (l25);
\draw[->] (l9) -- node [below, sloped] {\scriptsize{T = no}} (l26);
\draw[->] (l10) -- node [above, sloped] {\scriptsize{T = yes}} (l27);
\draw[->] (l10) -- node [below, sloped] {\scriptsize{T = no}} (l28);
\draw[->] (l11) -- node [above, sloped] {\scriptsize{T = yes}} (l29);
\draw[->] (l11) -- node [below, sloped] {\scriptsize{T = no}} (l30);
\draw[->] (l12) -- node [above, sloped] {\scriptsize{T = yes}} (l31);
\draw[->] (l12) -- node [below, sloped] {\scriptsize{T = no}} (l32);
\draw[->] (l13) -- node [above, sloped] {\scriptsize{T = yes}} (l33);
\draw[->] (l13) -- node [below, sloped] {\scriptsize{T = no}} (l34);
\draw[->] (l14) -- node [above, sloped] {\scriptsize{T = yes}} (l35);
\draw[->] (l14) -- node [below, sloped] {\scriptsize{T = no}} (l36);
\draw[->] (l15) -- node [above, sloped] {\scriptsize{T = yes}} (l37);
\draw[->] (l15) -- node [below, sloped] {\scriptsize{T = no}} (l38);
\draw[->] (l16) -- node [above, sloped] {\scriptsize{T = yes}} (l39);
\draw[->] (l16) -- node [below, sloped] {\scriptsize{T = no}} (l40);
\draw[->] (l17) -- node [above, sloped] {\scriptsize{T = yes}} (l41);
\draw[->] (l17) -- node [below, sloped] {\scriptsize{T = no}} (l42);
\draw[->] (l18) -- node [above, sloped] {\scriptsize{T = yes}} (l43);
\draw[->] (l18) -- node [below, sloped] {\scriptsize{T = no}} (l44);
\draw[->] (l19) -- node [above, sloped] {\scriptsize{T = yes}} (l45);
\draw[->] (l19) -- node [below, sloped] {\scriptsize{T = no}} (l46);
\draw[->] (l20) -- node [above, sloped] {\scriptsize{T = yes}} (l47);
\draw[->] (l20) -- node [below, sloped] {\scriptsize{T = no}} (l48);
\draw[->] (l21) -- node [above, sloped] {\scriptsize{T = yes}} (l49);
\draw[->] (l21) -- node [below, sloped] {\scriptsize{T = no}} (l50);
\draw[->] (l22) -- node [above, sloped] {\scriptsize{T = yes}} (l51);
\draw[->] (l22) -- node [below, sloped] {\scriptsize{T = no}} (l52);
\draw[->] (l23) -- node [above, sloped] {\scriptsize{T = yes}} (l53);
\draw[->] (l23) -- node [below, sloped] {\scriptsize{T = no}} (l54);
\draw[->] (l24) -- node [above, sloped] {\scriptsize{T = yes}} (l55);
\draw[->] (l24) -- node [below, sloped] {\scriptsize{T = no}} (l56);
\node at (-0.5,4) {
\begin{tabular}{cc}
     stage & $P($P = yes$| \cdot)$ \\
     \midrule
    \textcolor{purple}{purple} & $0.118$ \\
    \textcolor{orange}{orange} & $0.203$ \\
    \textcolor{green}{green} & $0.171$ \\
     \textcolor{blue}{blue} & $0.269$ \\
     \textcolor{cyan}{cyan} & $0.335$ \\
\end{tabular}
};
\node at (-0.5,-4) {
\begin{tabular}{cc}
     stage & $P($T = yes$| \cdot)$ \\
     \midrule
    \textcolor{brown}{brown} & $0.911$ \\
    \textcolor{red}{red} & $0.883$ \\
    \textcolor{yellow}{yellow} & $0.931$  \\ 
     \textcolor{cyan}{cyan} & $0.941$
\end{tabular}
};
\end{tikzpicture}
}

\caption{Staged tree maximizing the BIC for the order (F,E,S,P,T). \label{fig:staged3}}
\end{figure}
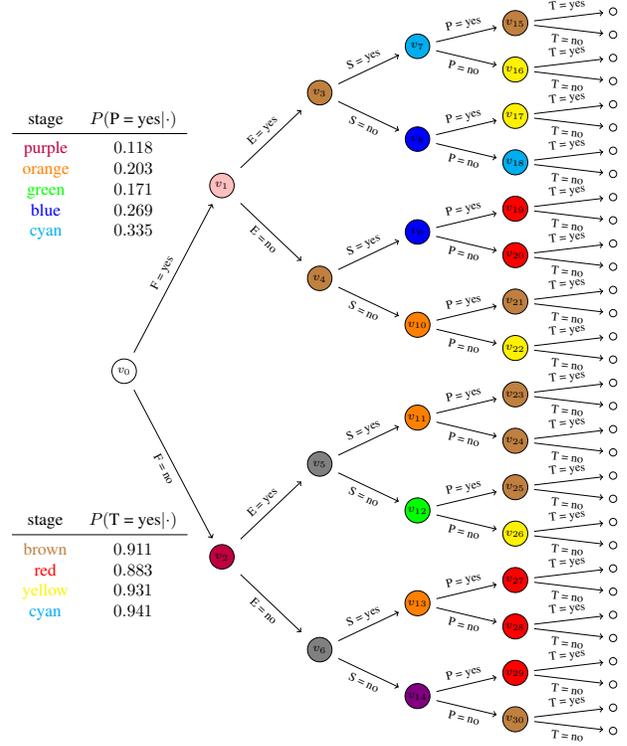

In order to understand in detail how the variables causally depend on each other, we can refer to the stages in Figure \ref{fig:staged3}. 
The staging over the variables P and T (vertices $v_7$ to $v_{30}$) shows an highly asymmetric dependence structure which could not be represented by a DAG model. For instance, the staging $\{v_{10},v_{11},v_{13}\}$ implies that, in the context S = yes and F = no, E has no causal effect on P. 
For the causal effects on watching TV (T), we can see that there are 
various contexts for which the trust on people (P) does not 
have an effect.
For example among people who practice sports regularly and 
are not satisfied with the environment of the area they live in. 
In the same context (S=yes, E=no) having friends they can count (F) on does also not 
appear to be a relevant factor for watching TV (T).
Moreover, for people who have friends they can rely on (F=yes), 
the probabilities of watching TV are the same if they do not practice sports (S=no) 
and they live in an unsatisfactory area (E=no), or they do
practice sports (S=yes) and they are satisfied with their area (E=yes).

It is apparent that the flexibility of the staging enables the intuitive representation of complex non-symmetric causal relationships learned from data.
In the Appendix, we further report the 
learned DAGs with different methods, interesting we observe that 
a variant of the PC algorithm~\citep{colombo14a} recovers a similar 
causal order while heuristics optimizing the BIC score are not 
able to infer any causal orderings of the variables.

\subsection{Outcomes for Hospitalised SARS-CoV-2 Patients}

We consider data on the trajectories of 
hospitalized SARS-CoV-2 patients in France
during the first nine months of the pandemic. 
In particular, we rely on the conditional 
probabilities reported by~\citet{LEFRANCQ2021} on the event that the hospitalized patient was transferred to ICU 
conditioned on gender, age, and on their death conditioned 
on gender, age, and if in ICU or not. 
Such probabilities were estimated by \citet{LEFRANCQ2021} from data on patients, recorded in the SI-VIC database, who started their hospitalization between 13 March and 30 November 2020.
Using those probabilities, we sampled 10000 artificial trajectories 
using the assumed true causal order (gender, age)$\to$ICU$\to$death.
We use the sampled trajectories to estimate a staged tree model using the BHC algorithm and the variable order search. 
The obtained staged tree model recovers the true causal order
and 
has a BIC score of 60421.77 while the DAG obtained with a \texttt{tabu}~\citep{russell2009artificial} 
search (optimizing BIC) obtains a higher BIC of 64227.09 and 
a complete DAG but the arc between gender and age.
Instead, the PC-stable algorithm~\citep{colombo14a} 
obtains a causal order similar to the one used in the data-generating
mechanism and the one retrieved by the staged tree.
We refer to  
the Appendix for the details on the  
learned staged tree, additional comments 
on the learned structure and comparisons with other DAG methods.

\subsection{ENSO Effects on Spring Precipitation in Australia}

We replicate here one of the examples described
by~\citet{QuantifyingCausalPathwaysofTeleconnections}. 
We consider, in particular, the causal inference question 
regarding the effect of El Niño Southern Oscillation (ENSO) 
on Australian precipitation (AU) during
spring, and the possible mediation of the 
Indian Ocean Dipole (IOD). 
As observed by 
\citet{QuantifyingCausalPathwaysofTeleconnections}:  
\textit{The influence of ENSO on the IOD, and thereby on AU, has been suggested to exhibit asymmetries in strength, implying that the relationship is nonlinear}. 
Instead of fitting a categorical DAG, we instead 
rely on staged tree models to capture and depict the 
asymmetric causal relationship in the data. 
The data are discretized following~\citet{QuantifyingCausalPathwaysofTeleconnections}: ENSO is reduced to three 
possible values (Niño, neutral, Niña), IOD into
three levels describing positive (+), neutral (0) and negative (-) phases; and AU is separated into
above (high) and below (low) average values.
We estimate the staged structure via the BHC algorithm (optimizing the AIC score) for the variable order (ENSO, IOD, AU), and we 
plot in Figure~\ref{fig:enso} the resulting staged tree (AIC$=368.23$) together with
the conditional probabilities of high AU 
in the three stages for the last variable. 
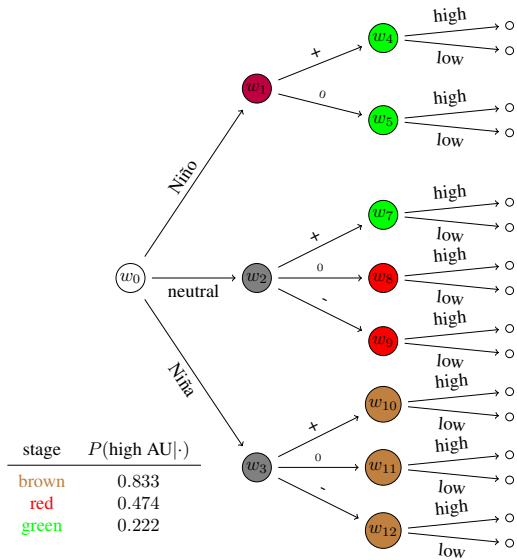
\begin{figure}
    \centering
\scalebox{0.7}{
\begin{tikzpicture}
\renewcommand{\xx}{2.4}
\renewcommand{\yy}{1.2}
\node (w0) at (0*\xx,0*\yy) {\stage{white}{0}{w}};
\node (w1) at (1*\xx,3*\yy) {\stage{purple}{1}{w}};
\node (w2) at (1*\xx,0*\yy) {\stage{gray}{2}{w}};
\node (w3) at (1*\xx,-3*\yy) {\stage{gray}{3}{w}};
\node (w4) at (2*\xx,3.8*\yy) {\stage{green}{4}{w}};
\node (w5) at (2*\xx,2.5*\yy) {\stage{green}{5}{w}};
\node (w7) at (2*\xx,1*\yy) {\stage{green}{7}{w}};
\node (w8) at (2*\xx,0*\yy) {\stage{red}{8}{w}};
\node (w9) at (2*\xx,-1*\yy) {\stage{red}{9}{w}};
\node (w10) at (2*\xx,-2*\yy) {\stage{brown}{10}{w}};
\node (w11) at (2*\xx,-3*\yy) {\stage{brown}{11}{w}};
\node (w12) at (2*\xx,-4*\yy) {\stage{brown}{12}{w}};
\node (l1) at (3*\xx,4*\yy) {\leaf};
\node (l2) at (3*\xx,3.6*\yy) {\leaf};
\node (l3) at (3*\xx,2.7*\yy) {\leaf};
\node (l4) at (3*\xx,2.3*\yy) {\leaf};
\node (l5) at (3*\xx,1.2*\yy) {\leaf};
\node (l6) at (3*\xx,0.8*\yy) {\leaf};
\node (l7) at (3*\xx,0.2*\yy) {\leaf};
\node (l8) at (3*\xx,-0.2*\yy) {\leaf};
\node (l9) at (3*\xx,-0.8*\yy) {\leaf};
\node (l10) at (3*\xx,-1.2*\yy) {\leaf};
\node (l11) at (3*\xx,-1.8*\yy) {\leaf};
\node (l12) at (3*\xx,-2.2*\yy) {\leaf};
\node (l13) at (3*\xx,-2.8*\yy) {\leaf};
\node (l14) at (3*\xx,-3.2*\yy) {\leaf};
\node (l15) at (3*\xx,-3.8*\yy) {\leaf};
\node (l16) at (3*\xx,-4.2*\yy) {\leaf};
\draw[->] (w0) -- node [above, sloped] {Niño} (w1);
\draw[->] (w0) -- node [below, sloped] {neutral} (w2);
\draw[->] (w0) -- node [below, sloped] {Niña} (w3);
\draw[->] (w1) -- node [above, sloped] {+} (w4);
\draw[->] (w1) -- node [above, sloped] {\tiny $0$} (w5);
\draw[->] (w2) -- node [above, sloped] {+} (w7);
\draw[->] (w2) -- node [above, sloped] {\tiny $0$} (w8);
\draw[->] (w2) -- node [above, sloped] {-} (w9);
\draw[->] (w3) -- node [above, sloped] {+} (w10);
\draw[->] (w3) -- node [above, sloped] {\tiny $0$} (w11);
\draw[->] (w3) -- node [above, sloped] {-} (w12);
\draw[->] (w4)--node [above, sloped] {high} (l1); \draw[->] (w4)--node [below, sloped] {low}(l2);
\draw[->] (w5)--node [above, sloped] {high}(l3); \draw[->] (w5) --node [below, sloped] {low}(l4);
\draw[->] (w7)--node [above, sloped] {high}(l5); \draw[->] (w7) --node [below, sloped] {low}(l6);
\draw[->] (w8)--node [above, sloped] {high}(l7); \draw[->] (w8) --node [below, sloped]{low}(l8);
\draw[->] (w9)--node [above, sloped] {high}(l9); \draw[->] (w9) --node [below, sloped]{low}(l10);
\draw[->](w10)--node[above, sloped]{high}(l11);\draw[->](w10)--node[below,sloped]{low}(l12);
\draw[->](w11)--node [above, sloped]{high}(l13);\draw[->] (w11)--node[below,sloped]{low}(l14);
\draw[->](w12)--node[above, sloped]{high}(l15);\draw[->](w12)--node[below,sloped]{low}(l16);
\node (tab) at (-0.5,-4) {
\begin{tabular}{cc}
     stage & $P($high AU$| \cdot)$ \\
     \midrule
    \textcolor{brown}{brown} & $0.833$ \\
    \textcolor{red}{red} & $0.474$ \\
     \textcolor{green}{green} & $0.222$ \\
    
\end{tabular}
};
\end{tikzpicture}
}
    \caption{Staged tree estimated with the BHC algorithm for the 
    ENSO-IOD-AU example and 
    estimated conditional probabilities for 
    high AU in the three recovered stages.}
    \label{fig:enso}
\end{figure}
We can observe that indeed there is an 
asymmetric relationship between ENSO, IOD 
and AU. 
In particular, AU does not depend on IOD 
in the extreme phases of ENSO (la Niña and el Niño) while the model suggests a negative correlation between IOD+ and high AU 
in the neutral ENSO phase.
These findings are consistent with the 
ones obtained 
by~\citet{QuantifyingCausalPathwaysofTeleconnections} by directly analyzing the contingency tables. 
Additionally, we can observe that the 
ENSO-IOD relationship seems to be asymmetric as well; from the stages of IOD we infer that 
$P($IOD$|$Niña$)=P($IOD$|$neutral$)$.
We conclude that in this example staged tree models allow a more intuitive and explainable analysis. 

As a further experiment we 
consider the alternative causal order 
IOD-ENSO-AU~\citep[as analyzed also by ][]{QuantifyingCausalPathwaysofTeleconnections}
and we estimate a staged tree model with the 
BHC algorithm. 
The obtained staged tree (see the Supplementary Material) has an AIC score of 371.04 and thus
the staged tree models suggest that the 
appropriate variables order is ENSO-IOD-AU.

\section{CONCLUSIONS}

We introduced and implemented causal discovery algorithms based on staged trees which extend classic DAG models to account for complex, non-symmetric causal relationships. In order to assess the effectiveness of staged trees in causal reasoning, we defined a new discrepancy that measures the agreement between the interventional distributions of two staged trees. Our simulation experiments demonstrate that if data is simulated from a staged tree model, and therefore embeds non-symmetric relationships between variables, staged trees outperform DAG models. Our real-world applications further highlight the need for non-symmetric models since staged trees, despite their complexity, outperform DAGs in terms of penalized fit and causal discovery.

We demonstrated that staged tree models can be
a valuable tool for causal discovery in real-world scenarios 
and various directions for future work are possible. 
We are currently focusing on the derivation of theoretical results 
about the identifiability of the causal order when non-symmetric relationships between two variables are present. 

Additional heuristics to learn the stage structures 
are currently being developed, and similarly 
different strategies for learning variable ordering. 
While the methods described in the present work obtain 
good results, they are lacking in scalability and 
new heuristics are needed to tackle a larger number of variables efficiently.

\subsubsection*{Acknowledgements}
Gherardo Varando’s work was funded by the European Research Council (ERC) Synergy
Grant “Understanding and Modelling the Earth System with Machine Learning (USMILE)”
under Grant Agreement No 855187.

\bibliographystyle{plainnat}
\bibliography{biblio_causal}

\appendix
\onecolumn

\section{Algorithm to Compute CID}

In the pseudo-code for the algorithm, we use the following notation: for a staged event tree $S=(W,F,\nu)$ we denote with
$\sim_{S}$ the equivalence relation over $V$ defined by $u \sim_S v$ if and only if $\nu(E(v)) = \nu(E(v))$.
Thus the equivalence classes, with respect to the above-defined
relation, are the stages of $S$.

\begin{algorithm}
\caption{Compute context interventional discrepancy} 
\label{alg:cid}
  \algsetup{linenosize=\tiny}
  \scriptsize
        \begin{algorithmic}
        \renewcommand{\algorithmiccomment}[1]{\hfill \# #1}
                \REQUIRE $T=(V,E,\eta)$ an $\bm{X}$-compatible staged event tree and
                         $S = (W,F,\nu)$ an $\bm{X}_{\pi}$-compatible 
                         staged event tree. 
                \ENSURE The CID between $T$ and $S$.   
       \STATE initialize $\texttt{CID} = 0$
       \FOR{$i = 1$ to $p$}       
              \STATE $k = \pi^{-1}(i)$ \COMMENT{the position of variable $X_i$ in staged tree $S$}
              \STATE $I = \{j: j < i \And \pi^{-1}(j) < k \} $
              \STATE $\mathcal{A} = \mathbb{X}_{\pi^{-1}([k-1])} / \sim_S $
               \COMMENT{the stages of $S$ at depth $k$}
               \STATE $\texttt{wrong} = \emptyset$
              \FOR {$A \in \mathcal{A}$}
                    \STATE $B_A = \{ \bm{x}_{[i-1]} \in \mathbb{X}_{[i-1]}: 
                             \bm{x}_I = \bm{y}_I \text{ for some }
                             \bm{y} \in A \}$
                    \IF{$|\eta(E(B_A))| > 1$} 
                       \STATE  $\texttt{wrong} = \texttt{wrong} \cup B_A$
                    \ENDIF
                    \STATE $\texttt{CID} = \texttt{CID} +  \frac{|\texttt{wromg}|}{|\mathbb{X}_{[i-1]}|}$ 
              \STATE 
              \ENDFOR
        \ENDFOR
        \end{algorithmic}
\end{algorithm} 

We here prove that Algorithm 1  correctly computes the CID as defined in Definition~\ref{def:cid}. Formally, we demonstrate that for any pair of $X$ and $X_{\pi}$ compatible staged event 
trees, $T$ and $S$, the output of Algorithm~\ref{alg:cid} is equal to 
$CID(T,S)$.

\proof
We need only to prove that the procedure in the first loop 
effectively identifies wrongly inferred interventional distributions 
of the type $P(X_i | \Do(\bm{X}_{[i-1]} = \bm{x}_{[i-1]}))$.
That is, at the end of $i$-th iteration of the 
first for loop, the variable $\texttt{wrong}$ contains the set of 
contexts $\bm{x}_{[i-1]} \in \mathbb{X}_{[i-1]}$ such that 
the corresponding interventional distribution for $X_i$ is wrongly 
inferred by $S$ with respect to $T$.
If $\bm{x}_{[i-1]} \in \texttt{wrong}$ then, by construction, 
there exists, a stage of $S$, $A \in \mathcal{A}$ such that $\bm{x}_{[i-1]} \in B_A$, and 
there exists $\bm{x}'_{[i-1]} \in B_A$ with
                             $\eta(E(\bm{x}_{[i-1]})) \neq \eta(E(\bm{x}'_{[i-1]}))$.
Thus, there exists $P \in \mathcal{M}_T$ such that 
$P(X_i| \bm{X}_{[t-1]} = \bm{x}_{[i-1]}) \neq P(X_i|\bm{X}_{[t-1]} = \bm{x}'_{[i-1]})$ and moreover 
\[ P(X_i| \bm{X}_{[i-1]} = \bm{x}_{[i-1]}) \neq P(X_i| \bm{X}_{[i-1]} \in B_A ).\]
Vice versa, if $\bm{x}_{[i-1]} \not\in \texttt{wrong}$,  we have that, since $\eta(E(B_A))$ is a singleton (all nodes in $B_A$ are in the same stage in $T$) thus, 
\[ P(X_i| \bm{X}_{[i-1]} = \bm{x}_{[i-1]}) = P(X_i| \bm{X}_{[i-1]} \in B_A ),\]
for every 
$P \in \mathcal{M}_T$ and for every $A \in \mathcal{A}$ such that $\bm{x}_{[i-1]} \in B_A$. 
And thus $S$ correctly infers $P(X_i|\bm{X}_{[i-1]} = \bm{x}_{[i-1]})$. 
\endproof

\section{Identifiability of the Causal Order} 

\label{sec:example}
As an instructive example, we report here a bivariate staged tree model compatible with  $(X_1, X_2)$ where the causal order 
can be identified by choosing the simpler model.

Consider the two staged trees depicted in Fig.~\ref{fig:bivcausal}. 
Let $T$ be the staged tree on the right and $S$ be the one on the left.  It is easy to see that if the data are generated from a joint probability 
distribution $P \in \mathcal{M}_T$, $S$ is the only $(X_1, X_2)$-compatible 
staged tree such that $P \in\mathcal{M}_S$. Indeed, since $\mathcal{M}_S$ is the saturated model we have that $\mathcal{M}_S$ is equal to the entire 
probability simplex. 

On the other hand, $\mathcal{M}_T \subsetneq \mathcal{M}_S$, since 
$P(X_1| X_2 = 2) = P(X_1|X_2 = 3)$ in $\mathcal{M}_T$, and thus 
the causal order is here identified by choosing the simplest model 
which describes the data-generating process. 
Even if this is a very simple example, it is  instructive to see that non-symmetrical conditional independence statements can be leveraged by
staged event trees to discover causal structure in categorical data.

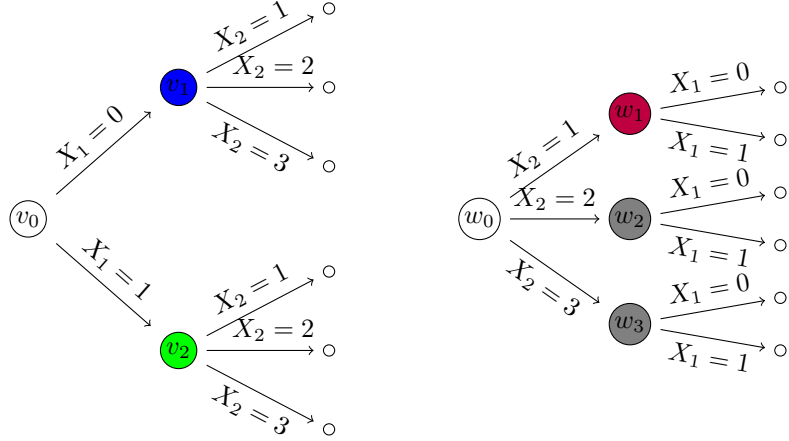
\begin{figure}[ht]
\centering
\begin{tikzpicture}
\renewcommand{\xx}{2}
\renewcommand{\yy}{0.7}
\node (v1) at (0*\xx,0*\yy) {\stage{white}{0}{v}};
\node (v2) at (1*\xx,2.5*\yy) {\stage{blue}{1}{v}};
\node (v3) at (1*\xx,-2.5*\yy) {\stage{green}{2}{v}};
\node (l1) at (2*\xx,4*\yy) {\leaf};
\node (l2) at (2*\xx,2.5*\yy) {\leaf};
\node (l3) at (2*\xx,1*\yy) {\leaf};
\node (l4) at (2*\xx,-1*\yy) {\leaf};
\node (l5) at (2*\xx,-2.5*\yy) {\leaf};
\node (l6) at (2*\xx,-4*\yy) {\leaf};
\draw[->] (v1) -- node [above, sloped] {$X_1=0$} (v2);
\draw[->] (v1) -- node [above, sloped] {$X_1=1$} (v3);
\draw[->] (v2) -- node [above, sloped] {$X_2=1$} (l1);
\draw[->] (v2) -- node [above, sloped] {$\quad X_2=2$} (l2);
\draw[->] (v2) -- node [below, sloped] {$X_2=3$} (l3);
\draw[->] (v3) -- node [above, sloped] {$X_2=1$} (l4);
\draw[->] (v3) -- node [above, sloped] {$\quad X_2=2$} (l5);
\draw[->] (v3) -- node [below, sloped] {$X_2=3$} (l6);

\node (w1) at (3*\xx,0*\yy) {\stage{white}{0}{w}};
\node (w2) at (4*\xx,2*\yy) {\stage{purple}{1}{w}};
\node (w3) at (4*\xx,0*\yy) {\stage{gray}{2}{w}};
\node (w4) at (4*\xx,-2*\yy) {\stage{gray}{3}{w}};
\node (l1) at (5*\xx,2.5*\yy) {\leaf};
\node (l2) at (5*\xx,1.5*\yy) {\leaf};
\node (l3) at (5*\xx,0.5*\yy) {\leaf};
\node (l4) at (5*\xx,-0.5*\yy) {\leaf};
\node (l5) at (5*\xx,-1.5*\yy) {\leaf};
\node (l6) at (5*\xx,-2.5*\yy) {\leaf};
\draw[->] (w1) -- node [above, sloped] {$X_2=1$} (w2);
\draw[->] (w1) -- node [above, sloped] {$X_2=2$} (w3);
\draw[->] (w1) --  node [below, sloped] {$X_2=3$} (w4);
\draw[->] (w2) -- node [above, sloped] {$X_1=0$} (l1);
\draw[->] (w2) -- node [below, sloped] {$X_1=1$} (l2);
\draw[->] (w3) -- node [above, sloped] {$X_1=0$} (l3);
\draw[->] (w3) -- node [below, sloped] {$X_1=1$} (l4);
\draw[->] (w4) -- node [above, sloped] {$X_1=0$} (l5);
\draw[->] (w4) -- node [below, sloped] {$X_1=1$} (l6);
\end{tikzpicture}

\caption{An example of an $(X_1, X_2)$-compatible (left) 
and an $(X_2, X_1)$-compatible (right) staged trees.
\label{fig:bivcausal}}
\end{figure}

\section{MISSING PROOFS}

\subsection{Proof of Proposition 1}

\proof

The first of the proposition 
follows from the observation that, 
imposing $G_T$ and $G'$ and  
the restricted sub-graphs (to total edges) 
$([p], F^{tot})$ and $([p], F')$  in the 
same Markov equivalence class implies that 
for each $j\in [p]$, all variables 
involved in non-asymmetric conditional 
independence statements with $X_j$, as well 
as all the other parents of $X_j$,
must appear before $X_j$ in every 
topological order of $G'$.
To prove that, consider 
$i < j \in [p]$ and 
assume that the edge $(i,j) \in F^{nt}$,
so by definition 
there exist contexts $\bm{x}_{[i-1]}$
and $\bm{x}'_{[i-1]}$, such that 
$\eta(E(\bm{x}_{[i-1]})) = 
\eta(E(\bm{x}'_{[i-1]}))$.
If we additionally assume $(k,j) \in G_{T}$,
for another $i\neq k < j$, we have that, 
\begin{itemize}
    \item If one of $(j,k)$ or $(k,j)$ 
    appears in $G_T$,
    then $(j,k)$ or $(k,j)$ must also be 
    in $G'$ (since they are Markov equivalent). 
    Similarly one of $(i,k)$ or $(k,i)$ must be in
    $G'$. Since $(i,j) \in F^{nt}$ and thus in 
    $G'$ by construction, then the only possibility is that $(k,j)$ is the 
    direction that appears in $G'$, otherwise
    either the acyclicity constrain or 
    the Markov equivalence between 
    $([p], F^{tot})$ and $([p], F')$ are violated.
    \item otherwise, the v-structure 
          $i \rightarrow j \leftarrow k$ is in 
          $G_T$ and thus in $G'$ (since they are Markow equivalent). 
\end{itemize}

Thus, summarizing, we have proved that 
if there is a non-total edge $(i,j)$
in $G_T$, the conditions on $G'$ imply
that all parents of $j$ in $G_T$ 
are also parents of $j$ in $G'$.
Let now be $\pi$ a topological order of $G'$;
it is easy to see that, we can build 
an $\bm{X}_{\pi}$-compatible staged tree $S$
such that   
$\mathcal{M}_T = \mathcal{M}_{S}$.

For the last statement, 
we prove the equivalent ``if $G_T$ and $G_{T'}$ are not Markov equivalent then $\mathcal{M}_T\neq \mathcal{M}_{T'}$". If $G_T$ and $G_{T'}$ are not Markov equivalent, it means that there is a conditional independence $X_A\independent X_B|X_C$ which is $\mathcal{M}_{G_T}$ but not in $\mathcal{M}_{G_{T'}}$, without loss of generality. However, the definition of minimal DAG implies that if $X_A\independent X_B | X_C$ is in $\mathcal{M}_{G_T}$ then it must also be in $\mathcal{M}_T$. Similarly, if $X_A\independent X_B | X_C$ is not in $\mathcal{M}_{G_{T'}}$ then it must not be in $\mathcal{M}_{T'}$. Therefore $\mathcal{M}_T\neq\mathcal{M}_{T'}$.
\endproof

As we have seen in the proof of the proposition
the conditions we impose on $G_T$ and $G'$ 
are very strong, and in fact, it is known that
in some cases, there are statistically equivalent 
staged trees whose minimal DAGs do not 
satisfy those assumptions. 
A more complete characterization of the 
equivalence classes of $\bm{X}_{\pi}$-compatible
staged trees needs probably to consider 
the different types of non-symmetric conditional
independences such as the ones discussed in
\citet{Varando2021}.

\subsection{Proof of Proposition 2}

\proof
If two staged tree $T,S$ are causally equivalent 
then for every $P \in \mathcal{M}_T = \mathcal{M}_S$ we have 
\[ P(X_i | \Do(\bm{X}_{[i-1]} = \bm{x}_{[i-1]})) =  P(X_i | \Do(\bm{X}_{I} = \bm{x}_{I})),   \]
thus,
\[ 
\begin{aligned}
P(X_i = x_i| \bm{X}_{[i-1]} = \bm{x}_{[i-1]}) &=    P(X_i = x_i | \Do(\bm{X}_{I} = \bm{x}_I))  \\
&= 
   P(X_i = x_i |  \bm{X}_I \in \{ \bm{y}_I \in \mathbb{X}_I: \nu(\bm{y}_K,w) = \nu(\bm{x}_K,u)  \} )
   \end{aligned}
\]
which proves point (i). 

To prove point (ii), observe that since $\pi$ is the identity, 
both $T = (V,E,\eta)$ and $S=(V,E,\nu)$ 
are $\bm{X}$-compatible staged event trees and thus 
$\mathcal{M}_T \subseteq \mathcal{M}_S$ implies that 
$\nu(u,v) = \nu(u', v') \Rightarrow \eta(u,v) = \eta(u', v')$
(the stage structure of $T$ is coarser than the one of $S$). 
Since $P \in \mathcal{M}_T$, we have,
\[
P(X_i = x_i| \bm{X}_{[i-1]} = 
\bm{x}_{[i-1]}) = 
   P(X_i = x_i | \bm{X}_{[i-1]} \in \{ \bm{y}_{[i-1]} \in \mathbb{X}_{[i-1]}: \nu(\bm{y}_{[i-1]},w) = \nu(\bm{x}_{[i-1]},u)  \} ).
\]

Finally, point (iii) follows from points (i) and (ii) by observing that if $T$ is the completely independent model then $T$ is causally equivalent to any completely independent $\bm{X}_{\pi}$-compatible staged tree, for any permutation $\pi$.
\endproof

\section{ADDITIONAL RESULTS}

\subsection{Simulation Experiment}

We report in Figures~\ref{fig:p3}, \ref{fig:p4} and \ref{fig:p5} the additional results for 
$p=3,4,5$ which could not fit in the main paper. 
Results show similar patterns to the case $p=6$ 
reported in the main paper.

Additionally, we plot the computational time 
for the four considered methods in Figure~\ref{fig:time}.

\begin{figure}
    \centering
    \includegraphics{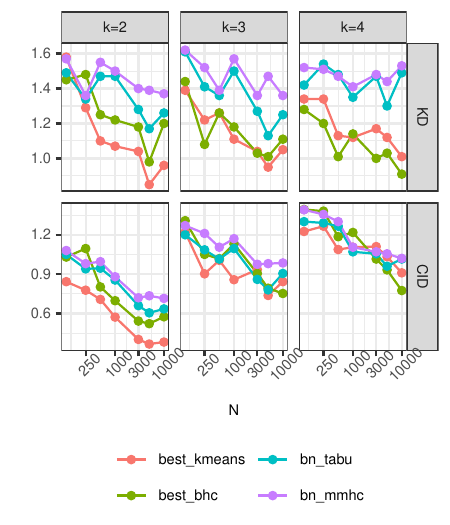}
    \caption{Results for the simulation experiments with $p=3$ binary variables.}
    \label{fig:p3}
\end{figure}

\begin{figure}
    \centering
    \includegraphics{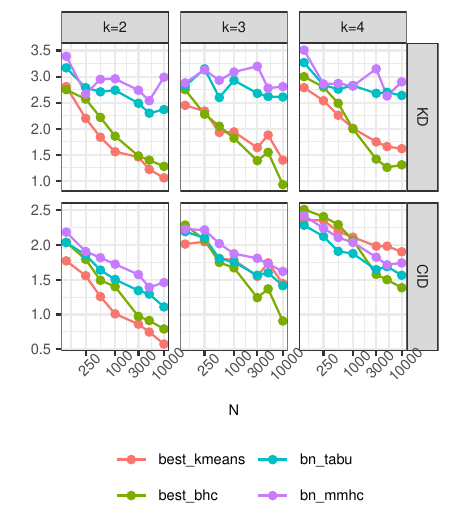}
    \caption{Results for the simulation experiments with $p=4$ binary variables.}
    \label{fig:p4}
\end{figure}

\begin{figure}
    \centering
    \includegraphics{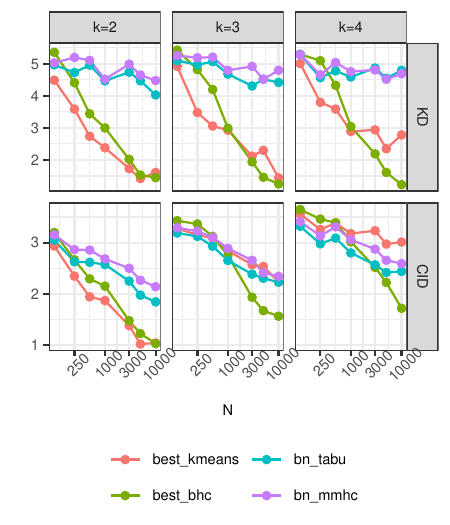}
    \caption{Results for the simulation experiments with $p=5$ binary variables.}
    \label{fig:p5}
\end{figure}

\begin{figure}
    \centering
    \includegraphics{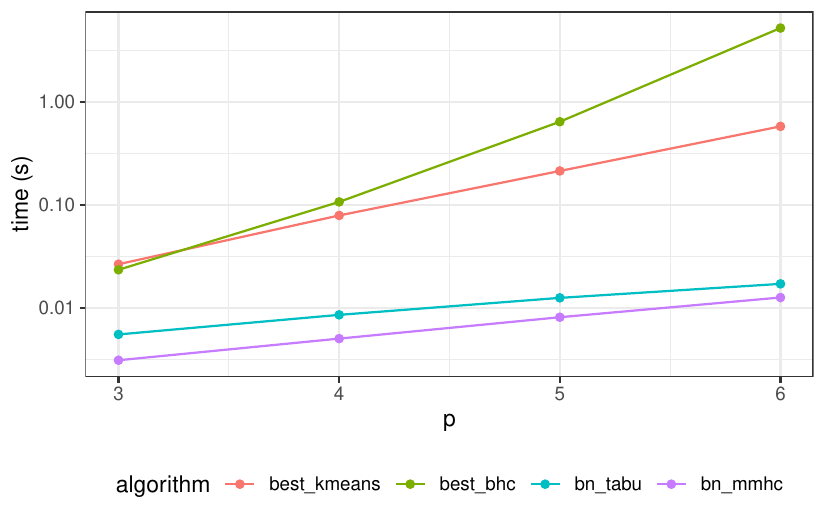}
    \caption{Results for the simulation experiments, computational time as a function of the number of variables $p$.}
    \label{fig:time}
\end{figure}

\subsection{Real World Examples}

\subsubsection{ISTAT: Aspects on Everyday Life}

For the ISTAT data on aspects of everyday life 
we fit standard DAG categorical models 
using two standard state-of-the-art methods: 
the PC algorithm with order-invariant implementation~\citep{colombo14a} and a 
hill-climbing search with tabu~\citep{russell2009artificial} list for 
optimizing the BIC score, both 
methods are available through the 
\texttt{bnlearn} package~\citep{Scutari2010}.
In Figure~\ref{fig:cpdagistat} the CPDAGs obtained 
with the two methods are reported. 
The CPDAGs are the unique representation of the DAG
Markov equivalence class. We can thus 
infer that both algorithms obtain a partial ordering 
(S,E)$\rightarrow$T, where the variable associated with 
watching television is estimated to be an effect of both
practicing sports and the satisfaction in the environment. 
Similar to the results obtained with the staged tree, 
the PC-stable algorithm identify also a partial ordering 
of the variables (S,F,E)$\rightarrow$ P while the 
ordering among S,F, and E cannot be inferred from data.
This results agree with the ones obtained through 
the staged tree models. 
The staged trees have the additional advantage of
depicting context-specific conditional independences.

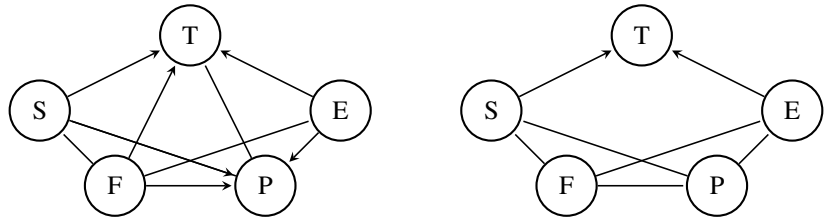
\begin{figure}
    \centering
\begin{tikzpicture}[
            > = stealth, 
            shorten > = 1pt, 
            auto,
            node distance = 1.5cm, 
            semithick 
        ]

        \tikzstyle{every state}=[
            draw = black,
            thick,
            fill = white,
            minimum size = 8mm
        ]

        \node[state] at(0,0) (1) {F};
        \node[state] at(2,0) (2) {P};
        \node[state] at(-1,1)  (3) {S};
        \node[state] at (3,1) (4){E};
        \node[state] at (1,2) (5) {T};
        
        \path[->] (1) edge   (2);
        \path[-] (1) edge   (3);
         \path[->] (1) edge   (5);
        \path[-] (1) edge   (4);
        \path[-] (2) edge   (3);
        \path[->] (4) edge   (2);
         \path[-] (2) edge   (5);
        \path[->] (3) edge (5);
        \path[->] (3) edge (2);
        \path[->] (4) edge (5);
        
 \end{tikzpicture}
 \hspace{1cm}
 \begin{tikzpicture}[
            > = stealth, 
            shorten > = 1pt, 
            auto,
            node distance = 1.5cm, 
            semithick 
        ]

        \tikzstyle{every state}=[
            draw = black,
            thick,
            fill = white,
            minimum size = 8mm
        ]

        \node[state] at(0,0) (1) {F};
        \node[state] at(2,0) (2) {P};
        \node[state] at(-1,1)  (3) {S};
        \node[state] at (3,1) (4){E};
        \node[state] at (1,2) (5) {T};
        
        \path[-] (1) edge   (2);
        \path[-] (1) edge   (3);
        \path[-] (1) edge   (4);
        \path[-] (2) edge   (3);
        \path[-] (2) edge   (4);
        \path[->] (3) edge (5);
        \path[->] (4) edge (5);
        
 \end{tikzpicture}
    \caption{CPDAGs for the ISTAT dataset, constructed with  PC-stable (left) and tabu search (right).}
    \label{fig:cpdagistat}
\end{figure}

\subsubsection{Outcomes for Hospitalised SARS-CoV-2 Patients}

We report additional figures and detail on the 
staged tree and DAGs learned for the data on trajectories of
hospitalized SARS-CoV-2 patients in France, during
the first nine months of the pandemic.

Data were obtained by simulations from a 
probability tree where conditional probabilities 
were obtained from \citet{LEFRANCQ2021}. 
In particular, conditional probabilities 
of ICU admission given age and gender and probabilities 
of death given ICU admission, age, and gender were obtained
from the tables on the supplementary materials 
provided by \citet{LEFRANCQ2021}.
Marginal probabilities of gender and probabilities
of age given gender were instead obtained from the 
linked GitHub repository.\footnote{
\url{https://github.com/noemielefrancq/Evolution-Outcomes-COVID19-France}}

The above conditional probabilities, define a 
probability event tree (a saturated staged tree model), 
with variables ordered gender, age, ICU, and death.
Simulations can be thus performed by iterative sampling 
of those variables. 
We simulate $10000$ trajectories and we use the 
artificial data to learn a staged tree model using the 
backward hill-climbing search coupled with 
the dynamic programming approach for 
variables order.

In Figure~\ref{fig:stcovid} we plot the learned staged tree,
we can appreciate that all nodes at depth one are in the same stage, thus the first two variables, age, and gender are
inferred to be independent, their causal order is thus 
not discernible from data (the age distributions given gender are indeed very similar).
For variables ICU and death, non-symmetrical and 
context-dependent conditional independences are presents.

In Figure~\ref{fig:cpdagcovid} we report
the CPDAG obtained with the PC-stable~\citep{colombo14a} 
and 
tabu algorithms~\citep{russell2009artificial}.
We observe that the score-based method (tabu) 
maximising BIC, is not able to find any 
ordering of the variables, except a 
conditional independence between gender and age.
On the contrary, the PC-stable results obtain 
a causal order similar to the one used in the 
data generating mechanism and the one 
retrieved by the staged tree.

\begin{figure}
\centering
\scalebox{0.6}{
\begin{tikzpicture}
\renewcommand{\xx}{2.2}
\renewcommand{\yy}{2.4}
\node (v0) at (0*\xx,0*\yy) {\stages{white}{0}};
\node (v1) at (1*\xx,3*\yy) {\stages{white}{1}};
\node (v2) at (1*\xx,1.8*\yy) {\stages{white}{2}};
\node (v3) at (1*\xx,0.6*\yy) {\stages{white}{3}};
\node (v4) at (1*\xx,-0.6*\yy) {\stages{white}{4}};
\node (v5) at (1*\xx,-1.8*\yy) {\stages{white}{5}};
\node (v6) at (1*\xx,-3*\yy) {\stages{white}{6}};
\node (v12) at (2*\xx,0.3*\yy) {\stages{gray}{12}};
\node (v11) at (2*\xx,0.9*\yy) {\stages{green}{11}};
\node (v10) at (2*\xx,1.5*\yy) {\stages{gray}{10}};
\node (v9) at (2*\xx,2.1*\yy) {\stages{red}{9}};
\node (v8) at (2*\xx,2.7*\yy) {\stages{blue}{8}};
\node (v7) at (2*\xx,3.3*\yy) {\stages{cyan}{7}};
\node (v13) at (2*\xx,-0.3*\yy) {\stages{green}{13}};
\node (v14) at (2*\xx,-0.9*\yy) {\stages{gray}{14}};
\node (v15) at (2*\xx,-1.5*\yy) {\stages{red}{15}};
\node (v16) at (2*\xx,-2.1*\yy) {\stages{gray}{16}};
\node (v17) at (2*\xx,-2.7*\yy) {\stages{gray}{17}};
\node (v18) at (2*\xx,-3.3*\yy) {\stages{gray}{18}};
\node (v19) at (3*\xx,3.45*\yy) {\stages{cyan}{19}};
\node (v20) at (3*\xx,3.15*\yy) {\stages{purple}{20}};
\node (v21) at (3*\xx,2.85*\yy) {\stages{white}{21}};
\node (v22) at (3*\xx,2.55*\yy) {\stages{purple}{22}};
\node (v23) at (3*\xx,2.25*\yy) {\stages{gray}{23}};
\node (v24) at (3*\xx,1.95*\yy) {\stages{cyan}{24}};
\node (v25) at (3*\xx,1.65*\yy) {\stages{gray}{25}};
\node (v26) at (3*\xx,1.35*\yy) {\stages{cyan}{26}};
\node (v27) at (3*\xx,1.05*\yy) {\stages{yellow}{27}};
\node (v28) at (3*\xx,0.75*\yy) {\stages{white}{28}};
\node (v29) at (3*\xx,0.45*\yy) {\stages{yellow}{29}};
\node (v30) at (3*\xx,0.15*\yy) {\stages{white}{30}};
\node (v31) at (3*\xx,-0.15*\yy) {\stages{yellow}{31}};
\node (v32) at (3*\xx,-0.45*\yy) {\stages{gray}{32}};
\node (v33) at (3*\xx,-0.75*\yy) {\stages{yellow}{33}};
\node (v34) at (3*\xx,-1.05*\yy) {\stages{white}{34}};
\node (v35) at (3*\xx,-1.35*\yy) {\stages{pink}{35}};
\node (v36) at (3*\xx,-1.65*\yy) {\stages{yellow}{36}};
\node (v37) at (3*\xx,-1.95*\yy) {\stages{pink}{37}};
\node (v38) at (3*\xx,-2.25*\yy) {\stages{gray}{38}};
\node (v39) at (3*\xx,-2.55*\yy) {\stages{pink}{39}};
\node (v40) at (3*\xx,-2.85*\yy) {\stages{pink}{40}};
\node (v41) at (3*\xx,-3.15*\yy) {\stages{pink}{41}};
\node (v42) at (3*\xx,-3.45*\yy) {\stages{gray}{42}};
\node (l1) at (4*\xx,3.5*\yy){\leaf};
\node (l2) at (4*\xx,3.4*\yy){\leaf};
\node (l3) at (4*\xx,3.2*\yy){\leaf};
\node (l4) at (4*\xx,3.1*\yy){\leaf};
\node (l5) at (4*\xx,2.9*\yy){\leaf};
\node (l6) at (4*\xx,2.8*\yy){\leaf};
\node (l7) at (4*\xx,2.6*\yy){\leaf};
\node (l8) at (4*\xx,2.5*\yy){\leaf};
\node (l9) at (4*\xx,2.3*\yy){\leaf};
\node (l10) at (4*\xx,2.2*\yy){\leaf};
\node (l11) at (4*\xx,2*\yy){\leaf};
\node (l12) at (4*\xx,1.9*\yy){\leaf};
\node (l13) at (4*\xx,1.7*\yy){\leaf};
\node (l14) at (4*\xx,1.6*\yy){\leaf};
\node (l15) at (4*\xx,1.4*\yy){\leaf};
\node (l16) at (4*\xx,1.3*\yy){\leaf};
\node (l17) at (4*\xx,1.1*\yy){\leaf};
\node (l18) at (4*\xx,1*\yy){\leaf};
\node (l19) at (4*\xx,0.8*\yy){\leaf};
\node (l20) at (4*\xx,0.7*\yy){\leaf};
\node (l21) at (4*\xx,0.5*\yy){\leaf};
\node (l22) at (4*\xx,0.4*\yy){\leaf};
\node (l23) at (4*\xx,0.2*\yy){\leaf};
\node (l24) at (4*\xx,0.1*\yy){\leaf};

\node (l48) at (4*\xx,-3.5*\yy){\leaf};
\node (l47) at (4*\xx,-3.4*\yy){\leaf};
\node (l46) at (4*\xx,-3.2*\yy){\leaf};
\node (l45) at (4*\xx,-3.1*\yy){\leaf};
\node (l44) at (4*\xx,-2.9*\yy){\leaf};
\node (l43) at (4*\xx,-2.8*\yy){\leaf};
\node (l42) at (4*\xx,-2.6*\yy){\leaf};
\node (l41) at (4*\xx,-2.5*\yy){\leaf};
\node (l40) at (4*\xx,-2.3*\yy){\leaf};
\node (l39) at (4*\xx,-2.2*\yy){\leaf};
\node (l38) at (4*\xx,-2*\yy){\leaf};
\node (l37) at (4*\xx,-1.9*\yy){\leaf};
\node (l36) at (4*\xx,-1.7*\yy){\leaf};
\node (l35) at (4*\xx,-1.6*\yy){\leaf};
\node (l34) at (4*\xx,-1.4*\yy){\leaf};
\node (l33) at (4*\xx,-1.3*\yy){\leaf};
\node (l32) at (4*\xx,-1.1*\yy){\leaf};
\node (l31) at (4*\xx,-1*\yy){\leaf};
\node (l30) at (4*\xx,-0.8*\yy){\leaf};
\node (l29) at (4*\xx,-0.7*\yy){\leaf};
\node (l28) at (4*\xx,-0.5*\yy){\leaf};
\node (l27) at (4*\xx,-0.4*\yy){\leaf};
\node (l26) at (4*\xx,-0.2*\yy){\leaf};
\node (l25) at (4*\xx,-0.1*\yy){\leaf};
\draw[->] (v0) --  node [above, sloped] {\scriptsize{80+}} (v1);
\draw[->] (v0) --  node [above, sloped] {\scriptsize{70-79}} (v2);
\draw[->] (v0) -- node [below, sloped] {\scriptsize{60-69}}(v3);
\draw[->] (v0) --  node [above, sloped] {\scriptsize{50-59}}(v4);
\draw[->] (v0) --  node [below, sloped] {\scriptsize{40-49}}(v5);
\draw[->] (v0) --  node [below, sloped] {\scriptsize{0-39}}(v6);
\draw[->] (v6) --  node [above, sloped] {\scriptsize{male}}(v17);
\draw[->] (v6) --  node [below, sloped] {\scriptsize{female}}(v18);
\draw[->] (v5) --  node [above, sloped] {\scriptsize{male}}(v15);
\draw[->] (v5) --  node [below, sloped] {\scriptsize{female}}(v16);
\draw[->] (v4) --  node [above, sloped] {\scriptsize{male}}(v13);
\draw[->] (v4) --  node [below, sloped] {\scriptsize{female}}(v14);
\draw[->] (v3) --  node [above, sloped] {\scriptsize{male}}(v11);
\draw[->] (v3) --  node [below, sloped] {\scriptsize{female}}(v12);
\draw[->] (v2) --  node [above, sloped] {\scriptsize{male}}(v9);
\draw[->] (v2) --  node [below, sloped] {\scriptsize{female}}(v10);
\draw[->] (v1) --  node [above, sloped] {\scriptsize{male}}(v7);
\draw[->] (v1) --  node [below, sloped] {\scriptsize{female}}(v8);
\draw[->] (v18) --  node [above, sloped] {\scriptsize{no-ICU}}(v41);
\draw[->] (v18) --  node [below, sloped] {\scriptsize{ICU}}(v42);
\draw[->] (v17) --  node [above, sloped] {\scriptsize{no-ICU}}(v39);
\draw[->] (v17) --  node [below, sloped] {\scriptsize{ICU}}(v40);
\draw[->] (v16) --  node [above, sloped] {\scriptsize{no-ICU}}(v37);
\draw[->] (v16) --  node [below, sloped] {\scriptsize{ICU}}(v38);
\draw[->] (v15) --  node [above, sloped] {\scriptsize{no-ICU}}(v35);
\draw[->] (v15) --  node [below, sloped] {\scriptsize{ICU}}(v36);
\draw[->] (v14) --  node [above, sloped] {\scriptsize{no-ICU}}(v33);
\draw[->] (v14) --  node [below, sloped] {\scriptsize{ICU}}(v34);
\draw[->] (v13) --  node [above, sloped] {\scriptsize{no-ICU}}(v31);
\draw[->] (v13) --  node [below, sloped] {\scriptsize{ICU}}(v32);
\draw[->] (v12) --  node [above, sloped] {\scriptsize{no-ICU}}(v29);
\draw[->] (v12) --  node [below, sloped] {\scriptsize{ICU}}(v30);
\draw[->] (v11) --  node [above, sloped] {\scriptsize{no-ICU}}(v27);
\draw[->] (v11) --  node [below, sloped] {\scriptsize{ICU}}(v28);
\draw[->] (v10) --  node [above, sloped] {\scriptsize{no-ICU}}(v25);
\draw[->] (v10) --  node [below, sloped] {\scriptsize{ICU}}(v26);
\draw[->] (v9) --  node [above, sloped] {\scriptsize{no-ICU}}(v23);
\draw[->] (v9) --  node [below, sloped] {\scriptsize{ICU}}(v24);
\draw[->] (v8) --  node [above, sloped] {\scriptsize{no-ICU}}(v21);
\draw[->] (v8) --  node [below, sloped] {\scriptsize{ICU}}(v22);
\draw[->] (v7) --  node [above, sloped] {\scriptsize{no-ICU}}(v19);
\draw[->] (v7) --  node [below, sloped] {\scriptsize{ICU}}(v20);
\foreach \n in {19,...,42}{
\draw[->] (v\n) --  node [above, sloped] {\scriptsize{survived}}(l\number\numexpr2*\n-38 + 1\relax);
\draw[->] (v\n) --  node [below, sloped] {\scriptsize{death}}(l\number\numexpr2*\n-38 + 2\relax);
}

\end{tikzpicture}
}
\caption{Staged tree obtained with the 
BHC method from the data on trajectories of covid patients. \label{fig:stcovid}}
\end{figure}
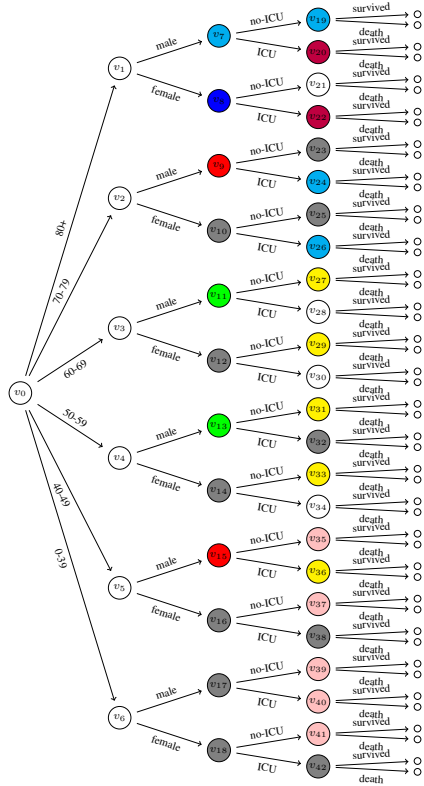

\begin{figure}
    \centering
\begin{tikzpicture}[
            > = stealth, 
            shorten > = 1pt, 
            auto,
            node distance = 1.5cm, 
            semithick 
        ]

        \tikzstyle{every state}=[
            draw = black,
            thick,
            fill = white,
            minimum size = 13mm
        ]

        \node[state] at(0,0) (1) {gender};
        \node[state] at(2,-2) (2) {age};
        \node[state] at(4,0)  (3) {ICU};
        \node[state] at (2,2) (4){death};
        
        \path[->] (1) edge   (4);
        \path[->] (1) edge   (3);
        \path[->] (2) edge   (4);
        \path[->] (2) edge   (3);
        \path[-] (4) edge   (3);

 \end{tikzpicture}
 \hspace{1cm}
 \begin{tikzpicture}[
            > = stealth, 
            shorten > = 1pt, 
            auto,
            node distance = 1.5cm, 
            semithick 
        ]

        \tikzstyle{every state}=[
            draw = black,
            thick,
            fill = white,
            minimum size = 13mm
        ]

        \node[state] at(0,0) (1) {gender};
        \node[state] at(2,-2) (2) {age};
        \node[state] at(4,0)  (3) {ICU};
        \node[state] at (2,2) (4){death};
        
        \path[-] (1) edge   (4);
        \path[-] (1) edge   (3);
        \path[-] (2) edge   (4);
        \path[-] (2) edge   (3);
        \path[-] (4) edge   (3);
        
 \end{tikzpicture}
    \caption{CPDAGs for the covid data, constructed with  PC-stable (left) and tabu search (right).}
    \label{fig:cpdagcovid}
\end{figure}
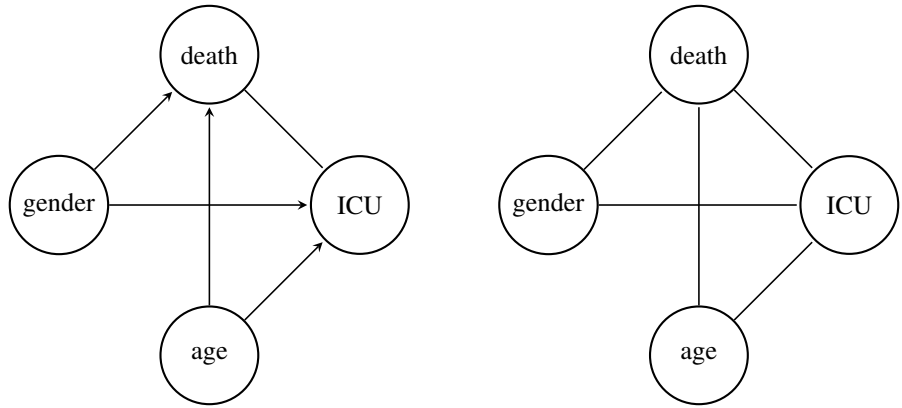

\subsubsection{ENSO Effects on Spring Precipitation in
Australia}

We continue here the analysis of the 
relationship between ENSO, IOD and spring 
precipitation in Australia (AU). 
In Figure~\ref{fig:stenso} we show the 
learned staged tree with the BHC algorithm with
the alternative order IOD, ENSO, AU.
The model attains an AIC of 371.04 (compared to
368.23 with the ENSO,IOD,AU order).
We can observe that the staging of the AU variable is the same (obviously with a node permutation in the plots), this fact is a consequence that 
the optimal staging of the nodes for 
a given variable (AU) is not dependent on the 
order of the previous variables (ENSO, IOD).

Thus the only real difference between the two staged trees is in the first two variables. 
We can see that for the ENSO$\rightarrow$IOD order, the method recognizes a so-called partial 
independence between ENSO and IOD 
$(P(\text{IOD}|\text{ENSO}= \text{Niña})=P(\text{IOD}|\text{ENSO}=\text{neutral})\neq P(\text{IOD}|\text{ENSO} = \text{Niño})$.
For the IOD$\rightarrow$ENSO model, instead 
no independence statement is found and the 
full model for the first two variables is 
obtained. Since the full model 
for ENSO and IOD could be written equivalently
with either of the two variables as the first one,
we can deduce that the ENSO$\rightarrow$IOD
ordering is thus supported by data under the assumption that the real model is the simpler one
that explains well the data.

\begin{figure}
    \centering
\scalebox{0.7}{
\begin{tikzpicture}
\renewcommand{\xx}{2.4}
\renewcommand{\yy}{1.2}
\node (w0) at (0*\xx,0*\yy) {\stage{white}{0}{w}};
\node (w1) at (1*\xx,3*\yy) {\stage{green}{1}{w}};
\node (w2) at (1*\xx,0*\yy) {\stage{gray}{2}{w}};
\node (w3) at (1*\xx,-3*\yy) {\stage{purple}{3}{w}};
\node (w4) at (2*\xx,3.8*\yy) {\stage{gray}{4}{w}};
\node (w5) at (2*\xx,2.5*\yy) {\stage{red}{5}{w}};
\node (w7) at (2*\xx,1*\yy) {\stage{gray}{7}{w}};
\node (w8) at (2*\xx,0*\yy) {\stage{red}{8}{w}};
\node (w9) at (2*\xx,-1*\yy) {\stage{yellow}{9}{w}};
\node (w10) at (2*\xx,-2*\yy) {\stage{gray}{10}{w}};
\node (w11) at (2*\xx,-3*\yy) {\stage{yellow}{11}{w}};
\node (w12) at (2*\xx,-4*\yy) {\stage{yellow}{12}{w}};
\node (l1) at (3*\xx,4*\yy) {\leaf};
\node (l2) at (3*\xx,3.6*\yy) {\leaf};
\node (l3) at (3*\xx,2.7*\yy) {\leaf};
\node (l4) at (3*\xx,2.3*\yy) {\leaf};
\node (l5) at (3*\xx,1.2*\yy) {\leaf};
\node (l6) at (3*\xx,0.8*\yy) {\leaf};
\node (l7) at (3*\xx,0.2*\yy) {\leaf};
\node (l8) at (3*\xx,-0.2*\yy) {\leaf};
\node (l9) at (3*\xx,-0.8*\yy) {\leaf};
\node (l10) at (3*\xx,-1.2*\yy) {\leaf};
\node (l11) at (3*\xx,-1.8*\yy) {\leaf};
\node (l12) at (3*\xx,-2.2*\yy) {\leaf};
\node (l13) at (3*\xx,-2.8*\yy) {\leaf};
\node (l14) at (3*\xx,-3.2*\yy) {\leaf};
\node (l15) at (3*\xx,-3.8*\yy) {\leaf};
\node (l16) at (3*\xx,-4.2*\yy) {\leaf};
\draw[->] (w0) -- node [above, sloped] {-} (w1);
\draw[->] (w0) -- node [below, sloped] {0} (w2);
\draw[->] (w0) -- node [below, sloped] {+} (w3);
\draw[->] (w1) -- node [above, sloped] {Niña} (w4);
\draw[->] (w1) -- node [above, sloped] {neutral} (w5);
\draw[->] (w2) -- node [above, sloped] {Niña} (w7);
\draw[->] (w2) -- node [above, sloped] {neutral} (w8);
\draw[->] (w2) -- node [above, sloped] {Niño} (w9);
\draw[->] (w3) -- node [above, sloped] {Niña} (w10);
\draw[->] (w3) -- node [above, sloped] {neutral} (w11);
\draw[->] (w3) -- node [above, sloped] {Niño} (w12);
\draw[->] (w4)--node [above, sloped] {high} (l1); \draw[->] (w4)--node [below, sloped] {low}(l2);
\draw[->] (w5)--node [above, sloped] {high}(l3); \draw[->] (w5) --node [below, sloped] {low}(l4);
\draw[->] (w7)--node [above, sloped] {high}(l5); \draw[->] (w7) --node [below, sloped] {low}(l6);
\draw[->] (w8)--node [above, sloped] {high}(l7); \draw[->] (w8) --node [below, sloped]{low}(l8);
\draw[->] (w9)--node [above, sloped] {high}(l9); \draw[->] (w9) --node [below, sloped]{low}(l10);
\draw[->](w10)--node[above, sloped]{high}(l11);\draw[->](w10)--node[below,sloped]{low}(l12);
\draw[->](w11)--node [above, sloped]{high}(l13);\draw[->] (w11)--node[below,sloped]{low}(l14);
\draw[->](w12)--node[above, sloped]{high}(l15);\draw[->](w12)--node[below,sloped]{low}(l16);
\end{tikzpicture}
}
    \caption{Staged tree estimated with the BHC algorithm for the 
    ENSO-IOD-AU example and 
    estimated conditional probabilities for 
    high AU in the three recovered stages.}
    \label{fig:stenso}
\end{figure}
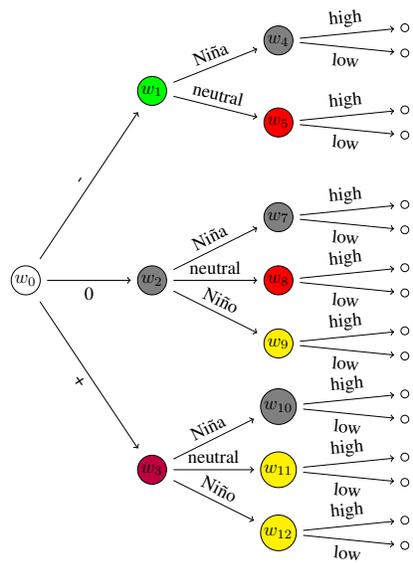

\vfill

\end{document}